\documentclass[11pt]{article}
\usepackage{color}
\usepackage[colorlinks=true, allcolors=blue,backref=page]{hyperref}
\usepackage{amsmath, amssymb, amsthm, thmtools}
\usepackage{mathrsfs}
\usepackage{mathtools}
\usepackage[noabbrev,capitalize,nameinlink]{cleveref}
\usepackage{fullpage}
\usepackage[noadjust]{cite}
\usepackage{graphics}
\usepackage{pifont}
\usepackage{tikz}
\usepackage{bbm}
\usepackage{float} 
\usepackage[T1]{fontenc}

\usetikzlibrary{arrows.meta}

\usepackage{environ}
\usepackage{framed}
\usepackage{url}
\usepackage[linesnumbered,ruled,vlined]{algorithm2e}
\usepackage[noend]{algpseudocode}
\usepackage[labelfont=bf]{caption}
\usepackage{cite}
\usepackage{framed}
\usepackage[framemethod=tikz]{mdframed}
\usepackage{appendix}
\usepackage{graphicx}
\usepackage[textsize=tiny,textwidth=2cm]{todonotes}
\usepackage{tcolorbox}
\usepackage{enumerate}
\allowdisplaybreaks[1]
\usepackage{enumerate}
\usepackage{stmaryrd}

\usepackage{wasysym}
\usepackage{MnSymbol}
\usepackage[margin=1in]{geometry}
\usepackage{braket}
\usepackage{todonotes}

\usepackage[shortlabels]{enumitem}
\crefformat{enumi}{#2#1#3}
\crefrangeformat{enumi}{#3#1#4 to~#5#2#6}
\crefmultiformat{enumi}{#2#1#3}
{ and~#2#1#3}{, #2#1#3}{ and~#2#1#3}

\DeclareSymbolFont{symbolsC}{U}{pxsyc}{m}{n}
\SetSymbolFont{symbolsC}{bold}{U}{pxsyc}{bx}{n}
\DeclareFontSubstitution{U}{pxsyc}{m}{n}
\DeclareMathSymbol{\medcircle}{\mathbin}{symbolsC}{7}

\crefname{equation}{}{} 
\AtBeginEnvironment{appendices}{\crefalias{section}{appendix}} 

\usepackage[color,final]{showkeys} 

\colorlet{refkey}{orange!20}
\colorlet{labelkey}{blue!30}

\numberwithin{equation}{section}
\newtheorem{theorem}{Theorem}[section]
\newtheorem{proposition}[theorem]{Proposition}
\newtheorem{lemma}[theorem]{Lemma}

\newtheorem{corollary}[theorem]{Corollary}
\newtheorem{conjecture}[theorem]{Conjecture}
\newtheorem*{question*}{Question}
\newtheorem{fact}[theorem]{Fact}

\theoremstyle{definition}
\newtheorem{definition}[theorem]{Definition}
\newtheorem{problem}[theorem]{Problem}
\newtheorem{question}[theorem]{Question}
\newtheorem*{definition*}{Definition}
\newtheorem*{theorem*}{Theorem}

\newtheorem{remark}[theorem]{Remark}

\newcommand{\mb}{\mathbb}
\newcommand{\mbf}{\mathbf}

\newcommand{\mc}{\mathcal}

\newcommand{\on}{\operatorname}
\newcommand{\wh}{\widehat}

\let\originalleft\left
\let\originalright\right
\renewcommand{\left}{\mathopen{}\mathclose\bgroup\originalleft}
\renewcommand{\right}{\aftergroup\egroup\originalright}

\newcommand{\mkw}[1]{%
\todo[color=green!20]{MW: #1}%
}
\newcommand{\mkwinline}[1]{%
\todo[inline,color=green!20]{MW: #1}%
}
\newcommand{\ks}[1]{%
\todo[color=yellow!20]{KS: #1}%
}

\renewcommand{\mkw}[1]{}
\renewcommand{\mkwinline}[1]{}
\renewcommand{\ks}[1]{}

\allowdisplaybreaks

\title{On Worst-Case Optimal Polynomial Intersection}
\author{Yihang Sun\thanks{Stanford University. \texttt{kimisun@berkeley.edu}}
  \and Mary Wootters\thanks{Stanford University. \texttt{marykw@stanford.edu}}}
\date{}
\begin{document}
\maketitle
\begin{abstract}
The \emph{Optimal Polynomial Intersection (OPI)} problem is the following: Given sets $S_1, \ldots, S_m \subseteq \mathbb{F}$ and evaluation points $a_1, \ldots, a_m \in \mathbb{F}$,  find a polynomial $Q \in \mathbb{F}[x]$ of degree less than $n$ so that $Q(a_i) \in S_i$ for as many $i \in \{1, 2, \ldots, m\}$ as possible.  \emph{Decoded Quantum Interferometry (DQI)} is a quantum algorithm that efficiently returns good solutions to the problem, even on worst-case instances~\cite{dqi}.  The quality of the solutions returned follows a \emph{semicircle law}, which outperforms known efficient classical algorithms.  But does DQI obtain the best possible solutions?  That is, are there solutions better than the semicircle law for worst-case OPI instances?  Surprisingly, before this work was first posted, the best existential results coincide with (and follow from) the best algorithmic results.

In this work, we show that there are better solutions for worst-case OPI instances over prime fields.  In particular, DQI and the semicircle law are not optimal.  For example, when the lists $S_i$ have size $\rho p$ for $\rho \sim 1/2$, our results imply the existence of a solution that asymptotically beats the semicircle law whenever $0.6225\le n/m<1$, and we show that an asymptotically perfect solution exists whenever  $n/m \geq 0.7496$.
Our results generalize to Max-LINSAT problems derived from any Maximum Distance Separable (MDS) code, and to any $\rho \in (0,1)$.
Along the way, we recover several re-proofs of the existence of solutions achieving the semicircle law.

While the results presented in this paper are existential, these ideas can lead to improved quantum algorithms: Since this work was first posted, subsequent work~\cite{HY26,Jo26} has given quantum algorithms attaining some of our existential results; as well as obtaining improved existential results.  The work \cite{Jo26} explicitly builds on our work by showing how to sample from a distribution that we construct to prove our existential results.

 The key insight to our improvement is a perhaps surprising connection to \emph{local leakage resilience of secret sharing schemes}. 
\end{abstract}
\maketitle
\section{Introduction}

The \emph{Optimal Polynomial Intersection} (OPI) problem is to find a low-degree polynomial that satisfies a large number of constraints on its evaluations.  Formally, we have the following definition.
\begin{problem}[Optimal Polynomial Intersection]\label{question:opi}
Fix prime power $q$ and integers $0\le n \le m \le q$, and fix distinct evaluation points $a_1, \ldots, a_m \in \mathbb{F}_q$. 
The \emph{Optimal Polynomial Intersection} (OPI) problem (with respect to $a_1, \ldots, a_m$) is the following: Given input subsets $S_1, \ldots, S_m$ with $S_i \subseteq \mb{F}_q$, find a polynomial $Q\in \mb{F}_q[x]$ with $\deg Q < n$ that maximizes the \emph{satisfaction ratio}
\[ s(Q) \coloneqq \frac{1}{m}|\{i \in \{1, 2, \ldots, m\} \,:\, Q(a_i) \in S_i \} |. \]
\end{problem}
OPI arises as a natural problem in many different areas.  In coding theory, it is related to \emph{list-recovery} of Reed-Solomon codes: There, the goal is to show that there are not too many polynomials $Q$ of degree less than $n$ so that $s(Q) \geq \alpha$ (for some parameter $\alpha$), and to return them all.  In cryptography, OPI has been studied under the name \emph{noisy polynomial reconstruction/interpolation} (e.g., \cite{np99,bn00}), and has been considered as a hardness assumption in certain parameter regimes.  More recently, OPI has arisen as a potential demonstration of quantum advantage: The \emph{Decoded Quantum Interferometry} (DQI) algorithm~\cite{dqi} and related algorithms~\cite{chailloux1,chailloux2,rosmanis2026nearly,khattar2025verifiable} give quantum algorithms to solve OPI that out-perform known efficient classical heuristics.

The guarantees of DQI for OPI hold in the worst case: For \emph{any} input sets $S_1, \ldots, S_m$, DQI finds a polynomial $Q$ of degree less than $n$ so that the fraction of satisfied constraints approaches a \emph{semicircle law}.  More precisely, for $\mu \in [0,1/2]$ and $\rho \in (0,1)$, 
let
\begin{equation}
\label{eq:scl}
\on{SCL}_\rho(\mu) \coloneqq \begin{cases}
\left(\sqrt{\mu(1-\rho)}+\sqrt{\rho(1-\mu)}\right)^2 &\text{if }\mu+\rho \le 1\\
1 & \text{if }\mu+\rho \ge 1
\end{cases}.
\end{equation}
Then, as $n,m \to \infty$ with fixed ratio $\mu = n/2m$, for any $S_1, \ldots, S_m \subseteq \mathbb{F}_q$ of size $\rho q$, \cite{dqi} shows that DQI efficiently finds a polynomial $Q$ of degree less than $n$ so that
\[ \mathbb{E}[s(Q)] \geq \on{SCL}_\rho(\mu) - o(1).\]
Above, the expectation is over the randomness of the algorithm.  When $\rho \sim 1/2$, we have 
\[ \on{SCL}_{1/2}(\mu) = \frac{1}{2} + \sqrt{\mu(1-\mu)},\]
the equation of a semicircle.  In contrast, the best known classical heuristic (Prange's algorithm) is only able to obtain a satisfaction ratio of $1/2 + \mu + o(1)$, even on average-case instances.

A natural question is how the semicircle law compares to the best possible solution to OPI.  Does DQI find optimal solutions in the worst case?  Before this work, it was not known if there were \emph{any} better solutions in the worst case.

\begin{question}\label{q:main}
Fix distinct $a_1, \ldots, a_m\in\mathbb{F}_q$.  Given $\rho \in (0,1)$ and $n, m \to \infty$ so that $n/2m \to \mu$, what is the optimal worst-case satisfaction ratio for OPI?  That is, what is
\[ \min_{\substack{S_1, \ldots, S_m \\ |S_i| = \rho q}} \max_{\substack{Q \in \mathbb{F}_q[x]\\\deg Q<n}} s(Q) \text{?} \] In particular, is it strictly larger than the semicircle law $\on{SCL}_\rho(\mu)$?  Or does DQI find asymptotically optimal worst-case solutions?
\end{question}

We note that prior work~\cite{chailloux1,chailloux2} obtained (algorithmic) improvements on the semicircle law for OPI; however those works consider the average case over an ensemble of OPI instances, not worst-case OPI instances as we do here.  We discuss the relationship to these works in \cref{sec:dqi-related}. 

We also note that, since our work was first posted on arXiv, subsequent work~\cite{HY26,Jo26} has improved our results existentially and also made some of them algorithmic, thereby improving on DQI algorithmically. We discuss these results more in \Cref{sec:subsequentwork}.

\paragraph{Our contributions.}  We highlight our main contributions below.
\begin{itemize}
\item Our main results, discussed quantitatively below in \cref{sec:results}, make progress on \cref{q:main}.  We show that, over prime fields, indeed there are solutions to worst-case OPI instances that are asymptotically better than the semicircle law $\on{SCL}_\rho(\mu)$, for a wide range of parameters $(\rho, \mu)$.
While our results are existential, we hope that our framework may lead to improved quantum algorithms.  We discuss this possibility more in \cref{sec:discussion}. 
\item Along the way, we develop new and simplified proofs of the existence of solutions on the semicircle law.  The fact that these solutions exist follows from the analysis of DQI~\cite{dqi}, but ``proof-by-quantum-algorithm'' does not seem like the ``correct'' way to prove this classical combinatorial statement.  We give two simpler re-proofs of this result.  These simpler proofs form the basis of our (more complicated) improvements.  We note that some of the ideas in our re-proofs are present in \cite{tie, dqi-complexity}; by making them explicit, we are able to identify avenues for improvement.   
\item Finally, we establish a connection between OPI and \emph{leakage-resilient secret sharing}, discussed more in \Cref{sec:LRS_connection}.  We use and improve techniques from leakage-resilient secret sharing to establish our results about OPI.  While our improvements do not immediately lead to improvements in the best known results in leakage-resilient secret sharing, they do lead to improvements among a class of results using a particular set of techniques, and we hope that this connection may lead to absolute improvements in the future.
\end{itemize}

More generally, our results apply to any Max-LINSAT problem that arises from a \emph{Maximum Distance Separable} (MDS) code. A \emph{linear code} $C \subseteq \mathbb{F}_q^m$ of dimension $n$ is just an $n$-dimensional subspace of $\mathbb{F}_q^m$.  A matrix $B \in \mathbb{F}_q^{m \times n}$ is a \emph{generator matrix} of $C$ if $C$ is the column span of $B$.  We say that $C$ is \emph{Maximum Distance Separable} (MDS) if any $n$ rows of $B$ are linearly independent; equivalently, if any $n$ symbols of a codeword $c \in C$ determine the entire codeword.

\begin{problem}[(MDS) Max-LINSAT]
\label{question:max-linsat} Fix a prime power $q$ and a matrix $B \in \mathbb{F}_q^{m \times n}$.  The \emph{Max-LINSAT} problem (with respect to $B$) is the following: Given input lists $S_1, \ldots, S_m$ with $S_i \subseteq \mb{F}_q$,
find $x\in \mb{F}_q^{n}$ that maximizes the \emph{satisfaction ratio}
\begin{equation}
s(x)\coloneqq \frac{1}{m}\left|\left\{i\in\{1,2, \ldots, m\}:(Bx)_i\in S_i \right\}\right|.
\end{equation}
When the matrix $B$ is the generator matrix of an MDS code, we call the problem  \emph{MDS Max-LINSAT} (with respect to $B$).
\end{problem}
As observed in \cite[Section 5]{dqi}, OPI is a special case of MDS Max-LINSAT when the matrix $B$ is the Vandermonde matrix with $B_{i,j} = a_i^{j-1}$.  It corresponds to the case when the MDS code $C$ is a \emph{Reed-Solomon} code.\footnote{The Reed--Solomon code consists of the vectors $(Q(a_1),\ldots,Q(a_m))$ as $Q$ ranges over polynomials in $\mathbb{F}_q[x]$ of degree less than $n$.}
We note that DQI applies to more general Max-LINSAT instances, although OPI is of special interest because of its potential demonstration of significant quantum advantage.

We remark that all of our results hold for the general MDS Max-LINSAT problem for \emph{any} matrix $B \in \mb{F}_q^{m \times n}$ that is the generator matrix of an MDS code.  When restricted to OPI, this means that our results hold for any choice of evaluation points $a_i$.  Thus, we omit the dependence on $B$ in the notation $s(x)$, and we will often refer the ``the'' MDS Max-LINSAT problem, rather than the MDS Max-LINSAT problem with respect to a particular matrix $B$.  Similarly, we omit the dependence of the evaluation points $a_i$ on the notation $s(Q)$, and refer the ``the'' OPI problem rather than OPI with respect to particular evaluation points.
 We parameterize both OPI and Max-LINSAT in terms of $n/m = 2\mu$, which is the rate of the code whose generator matrix is $B$.  

\subsection{Main Results}\label{sec:results}
Our main results hold over prime fields $\mb{F}_p$.  We consider the asymptotic regime where $m, n\to\infty$ with fixed rate $2\mu = n/m$, and the setting where the sets $S_i$ have size $\rho p$, for some $\rho \in (0,1)$. 
As we present our results, 
we highlight two {threshold rates} for $\mu$:
\begin{itemize}
    \item \emph{Improvement Threshold:} Let $\mu_0(\rho)$ be the minimum $\mu$ for which DQI is not asymptotically optimal for MDS Max-LINSAT for any input lists of size $\rho p$, i.e. for any sets $S_1, \ldots, S_m$ of size $\rho p$, there exists an $\varepsilon>0$ and a solution $x \in \mathbb{F}_p^n$ such that $s(x)\ge \on{SCL}_\rho(\mu)+\varepsilon-o(1)$.
    \item \emph{Saturation Threshold:} Let $\mu_1(\rho)$ be the minimum $\mu$ for which an asymptotically perfect solution exists to MDS Max-LINSAT for any input lists of size $\rho p$, i.e. for any sets $S_1, \ldots, S_m$ of size $\rho p$, there exists a solution $x \in \mathbb{F}_p^n$ such that $s(x)\ge 1-o(1)$. 
\end{itemize}
We note that the saturation threshold obeys monotonicity: If $\mu>\mu_1(\rho)$, then an asymptotically perfect solution exists for $\mu$. A priori, this is not true for the improvement threshold, 
but it does hold for our results: Fix any $\rho$, if we can improve on the semicircle law for the case of $\mu_0$, then we can do so for any $\mu \in [\mu_0, 1-\rho)$ in the allowed range.
Let $H$ be the binary entropy function
\begin{equation}\label{eq:Hx}
    H(x)\coloneqq -x\log x-(1-x)\log (1-x).
\end{equation}
We begin with the balanced case where $\rho = |S_i|/p\sim 1/2$. 
Below, we state two results, \cref{thm:main-avg} and \cref{thm:main-best}.  The second theorem is quantitatively stronger, but is more complicated to state.  Both results are plotted in \cref{fig:balanced}, along with the bounds they imply on the Improvement and Saturation Thresholds.

\begin{figure}[H]
    \centering
    \includegraphics[width=0.9\textwidth]{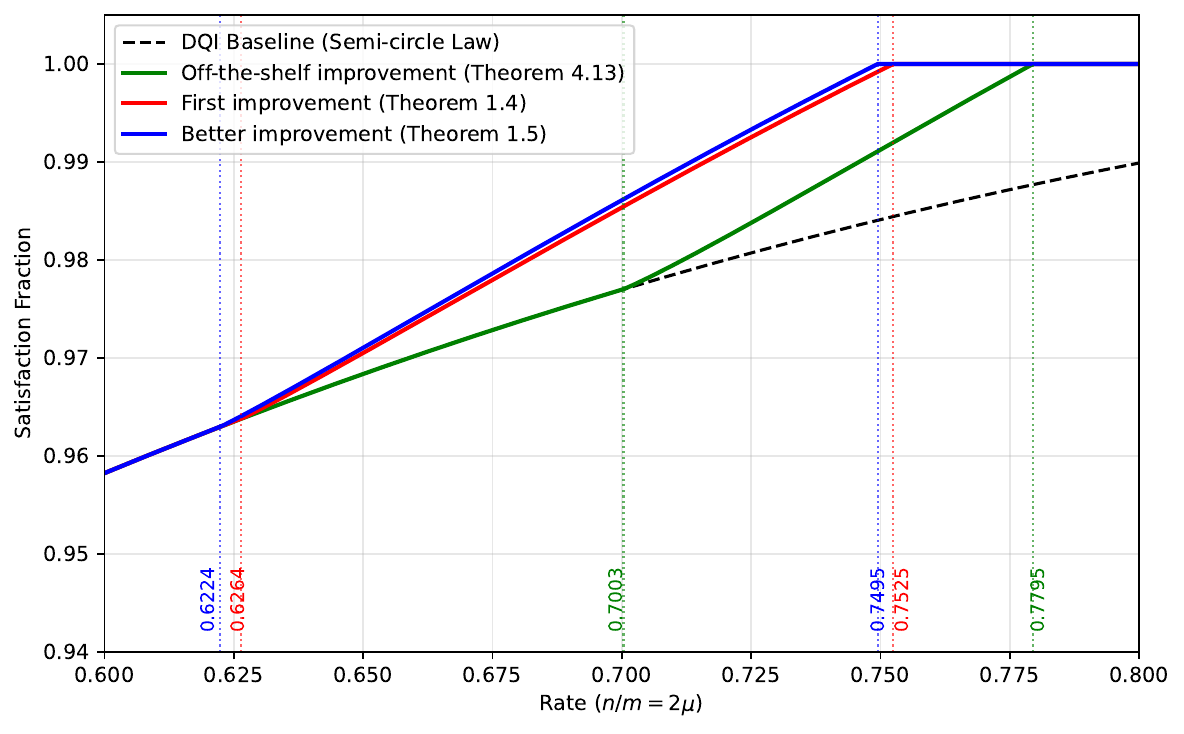}
    \caption{Three improvements on the semicircle law in the balanced case where $\rho\sim 1/2$, including \cref{thm:main-avg,thm:main-best}.  The green curve is the result of applying techniques from leakage-resilient secret-sharing in an off-the-shelf way, which we discuss more in \cref{sec:overview}. We color code and mark the bounds on critical rates of each curve: $2\mu_0(\rho)$ is where each curve diverges from $\on{SCL}_{1/2}(\mu)$, and $2\mu_1(\rho)$ when it hits $1$ asymptotically.} 
    \label{fig:balanced}
\end{figure}

\begin{restatable}{theorem}{thmA}\emph{(First improvement; balanced case)}
\label{thm:main-avg}
Let $2\mu \in [0,1]$, and let $n,m$ be sufficiently large, with $n/m = 2\mu$. Let $p$ be prime and fix any MDS generator matrix $B \in \mb{F}_p^{m \times n}$.  Then for any input sets $S_i$ of size $|S_i|/p \sim 1/2$, the MDS Max-LINSAT problem (with respect to $B$)
admits a solution $x \in \mathbb{F}_p^n$ with satisfaction ratio
\begin{equation}
 s(x)\ge \on{SCL}_{1/2}(\mu+\delta)-o(1),
\end{equation}
for any $\delta \in [0, 1/2-\mu]$ such that $E(\mu, \delta)+F(\mu)< 0$, where
\begin{equation}
\label{eq:EFH}
\begin{aligned}
E(\mu, \delta) & \coloneqq(\mu+\delta)H\left(\frac{\mu}{\mu+\delta}\right)+(1-\mu-\delta)H\left(\frac{\mu}{1-\mu-\delta}\right)-H(2\mu)\\
F(\mu) & \coloneqq (1-2\mu)\log 2 +2\mu\left(4\mu-1\right) \log\left(\frac{2}{\pi}\right).
\end{aligned}
\end{equation}
\end{restatable}

As $\on{SCL}(\mu + \delta)$ is increasing with $\delta$, we want the largest feasible $\delta>0$; call this $\delta_{\max}$. In \cref{fig:balanced}, we plot $\on{SCL}_{1/2}(\mu+\delta_{\max})$ in red, improving the DQI benchmark $\on{SCL}_{1/2}(\mu)$ in the black dashed line. We unpack the thresholds: 
\begin{itemize}
    \item \emph{Improvement Threshold:} When no such $\delta$ exists, we cannot improve on the semicircle law; otherwise, when any $\delta >0$ exists so that $E(\mu, \delta)+F(\mu)<0$, we can improve asymptotically from the semicircle law. From \cref{fig:balanced}, we see that \cref{thm:main-avg} implies $2\mu_0(1/2)<0.6265$. 
    \item \emph{Saturation Threshold:} Since $\on{SCL}_{1/2}(1/2)=1$, the largest possible $\delta$ we would ever take is $1/2-\mu$, so feasibility of $\delta=1/2-\mu$ bounds $\mu_1(1/2)$.
    From \cref{fig:balanced}, we see that \cref{thm:main-avg} implies $2\mu_1(1/2)< 0.7526$. 
\end{itemize}
We further improve on this result as follows, at the expense of additional complexity in the statement.
\begin{restatable}{theorem}{thmB}\emph{(Better improvement; balanced case)}
\label{thm:main-best}
Let $2\mu \in [0,1]$, and let $n,m$ be sufficiently large, with $n/m = 2\mu$.  Let $p$ be prime and fix any MDS generator matrix $B \in \mb{F}_p^{m \times n}$.  Then for any input sets $S_i$ with size $|S_i|/p \sim 1/2$, the MDS Max-LINSAT problem (with respect to $B$) admits a solution $x \in \mathbb{F}_p^n$
 with satisfaction ratio
\begin{equation}
 s(x)\ge \on{SCL}_{1/2}(\mu+\delta)-o(1),
\end{equation}
for any $\delta\in[0,1/2-\mu]$ and $\lambda\in[\max\{0,6\mu-2\},4\mu-1]$ such that $E(\mu,\delta)+G(\mu,\lambda)<0$, where
\begin{equation}
\label{eq:EFG}
\begin{aligned}
E(\mu, \delta) & \coloneqq(\mu+\delta)H\left(\frac{\mu}{\mu+\delta}\right)+(1-\mu-\delta)H\left(\frac{\mu}{1-\mu-\delta}\right)-H(2\mu),\\
G(\mu, \lambda ) & \coloneqq \left(1-2\mu\right)\log 2 +H\left(2\mu \right) -\left(4\mu -1\right)H\left(\frac{\lambda }{4\mu -1}\right)-\left(2-4\mu \right)H\left(\frac{2\mu -\lambda }{2-4\mu }\right)+\lambda \log\left(\frac{2}{\pi}\right).
\end{aligned}
\end{equation}
Solving for minimizer $\lambda_\star$ of $G(\mu, \lambda)$ gives
\begin{equation}
\lambda_\star\coloneqq \frac{A-\sqrt{A^2-8(1-2/\pi)\mu(4\mu-1)}}{2(1-2/\pi)}\quad\text{where}\quad A\coloneqq \frac{2}{\pi}+\left(1-\frac{2}{\pi}\right)(6\mu-1).
\end{equation}
\end{restatable}

In \cref{fig:balanced}, this is plotted as the blue line, which slightly improves on \cref{thm:main-avg}. This improves the bounds to the threshold rates to $2\mu_0(1/2)< 0.6225$ and $2\mu_1(1/2) < 0.7496$.
We record these observations below as a direct answer to \cref{q:main}.
\begin{corollary}
\label{cor:feature-cor}
For MDS Max-LINSAT with $|S_i|/p \sim 1/2$, DQI does \emph{not} find the asymptotically optimal solution if the rate satisfies $0.6225\le 2\mu<1$. Moreover, an asymptotically perfect solution exists for rates $2\mu\ge 0.7496$.
\end{corollary}

\begin{remark}\label{rem:chailloux}
   The improvement of $2 \mu_1(1/2) < 0.7496$ is notable as $0.7496 < 3/4$.  This $3/4$ bound is the rate above which \cite{chailloux2} can \emph{algorithmically} find an asymptotically perfect solution in expectation, over a slightly randomized ensemble of input lists $S_i$ (see the discussion in \cref{sec:dqi-related}).   
\end{remark}

We also obtain similar results for a general $\rho\in (0,1)$ below. The proof of \cref{thm:main-biased} is analogous to that of \cref{thm:main-avg} for the balanced $\rho \sim 1/2$ case. It is possible to slightly improve \cref{thm:main-biased} using ideas from the proof of~\cref{thm:main-best}, but we pursue that direction only for the balanced case for simplicity. 
\begin{restatable}{theorem}{thmC}\label{thm:main-biased}\emph{(Main theorem; biased case)}
Let $2\mu \in [0,1]$, and let $n,m$ be sufficiently large, with $n/m = 2\mu$.  Let $\rho \in (0,1)$. Let $p$ be prime, and fix an MDS generator matrix $B \in \mb{F}_p^{m \times n}$.  Then for any input sets $S_i$ with size $|S_i| = \rho p$, the MDS Max-LINSAT problem (with respect to $B$) admits a solution $x \in \mathbb{F}_p^n$ with satisfaction ratio
\begin{equation}
 s(x)\ge \on{SCL}_\rho(\mu+\delta)-o(1),
\end{equation}
for any $\delta \in [0, 1-\rho-\mu]$ such that
$E_\rho(\mu,\delta)+F_\rho(\mu)<0$, where
\begin{equation}
\begin{aligned}
E_\rho(\mu, \delta) & \coloneqq 2\mu \log 2-H(\mu+\delta)+
\max_{\gamma} \left\{\gamma\log\left(\frac{|1-2\rho|}{2\sqrt{\rho(1-\rho)}}\right)+2\mu H\left(\frac{\gamma}{2\mu}\right)+(1-2\mu) H\left(\frac{\delta-\gamma/2}{1-2\mu}\right)\right\},
\\ F_\rho \left(\mu\right) & \coloneqq  (2\mu-1)\log\rho +\mu\log \left(\frac{\rho}{1-\rho}\right)+2\mu(4\mu-1)\log \left(\frac{|\sin(\rho \pi)|}{\rho \pi}\right).
\end{aligned}
\end{equation}
\end{restatable}
\begin{figure}[H]
    \centering
 \includegraphics[width=0.8\textwidth]{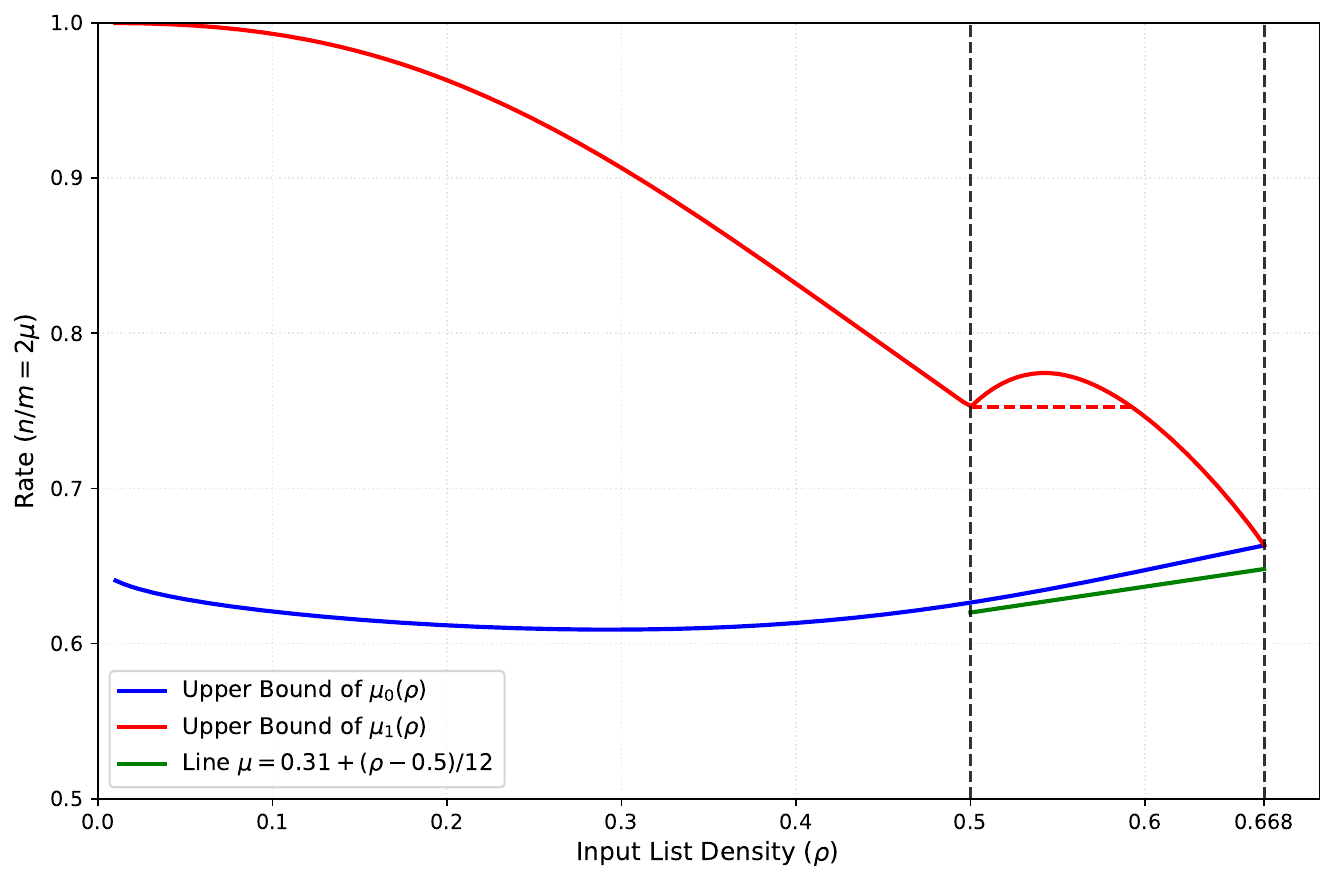}
    \caption{Numerical phase diagram for the upper bounds from \cref{thm:main-biased} on the improvement and saturation thresholds as functions of the list density $\rho$. The dashed red segment uses monotonicity of the saturation threshold. The green segment shows the numerical rate restriction used in \cref{sec:general-proofs}.}
    \label{fig:phase}
\end{figure}

In \cref{fig:phase}, we plot numerical upper bounds on threshold rates $2\mu_0(\rho)$ and $2\mu_1(\rho)$. The numerical evaluation shows improvement up to a density of about $0.668$.

The bump right after $1/2$ is an artifact of our analysis, and stems from having absolute values around $1-2\rho$ in $E_\rho$ (see \cref{rem:tildeN}). Since the worst-case maximum satisfaction fraction is clearly increasing in $\rho$, we see that the saturation threshold $\mu_1(\rho)$ is decreasing in $\rho$. Thus, we can replace the bump with the dashed horizontal red segment.
Finally, we have included figures (\cref{fig:biased}) showcasing improvement over the biased semicircle law for typical values of $\rho$ in \cref{sec:general-proofs}. 

\subsection{Technical Overview}
\label{sec:overview}
In this section, we outline our approach.  For simplicity, we focus on the perfectly balanced case where $\rho = 1/2$.\footnote{We note that it is not possible for $\rho$ to be exactly $1/2$ and for the field size to be a large prime (which is required for our main results, \cref{thm:main-avg,thm:main-best,thm:main-biased}), as $\rho |\mathbb{F}|$ must be an integer.  Generalizing to arbitrary $\rho$ (either $1/2 + o(1)$ or more generally) requires some delicate analysis, but does not change the main ideas of the proof.  We discuss this briefly at the end of this section.}  
%

\subsubsection{Two re-proofs of the semicircle law, and what we learn from them}\label{sec:overview1}
Our starting point is the observation that the semicircle law attained by DQI in \cite{dqi} is a lower bound on the worst-case satisfaction ratio.

\begin{proposition}[\cite{dqi}]
\label{prop:dqi-on-opi}
The MDS Max-LINSAT problem for any sets $S_i$ of size $\rho q$ and rate $2\mu= n/m$ admits a solution $x \in \mathbb{F}_q^n$ with satisfaction ratio $s(x)\ge \on{SCL}_{\rho}(\mu)-o(1)$.
\end{proposition}

A natural question is whether there is a classical, combinatorial proof of this classical, combinatorial statement (rather than a proof-by-quantum-algorithm).
The answer is yes, and 
we begin in \cref{sec:re-proof} with two re-proofs of \cref{prop:dqi-on-opi}, for the special case that $\rho = 1/2$.\footnote{These proofs can be generalized to $\rho \in (0,1)$ with ideas from \cref{sec:improvements}, but since the point of these re-proofs is their relative simplicity, we consider only the $\rho = 1/2$ case in \cref{sec:re-proof}.}  We describe these below.  

In \cref{sec:moments-proof}, we give a very succinct re-proof of \cref{prop:dqi-on-opi} using known results on the \emph{moments problem} (see \cref{sec:moments}).  This proof bypasses the Fourier analysis and coding theoretic arguments in \cite{dqi}. The key idea is that for a uniform random solution $x \in \mathbb{F}_q^n$, $s(x)$ has the same first $n$ moments as $Z/m$ where $Z\sim \on{Bin}(m, 1/2)$; this follows from the fact that the underlying matrix $B$ is the generator matrix of an MDS code.  
Known results on the {moments problem} then say something about the support of the random variable $s(x)$ (over the randomness of $x \in \mathbb{F}_q^n$).  In more detail, \cref{thm:cms} implies that the support of any distribution with these moments must interlace the support of a distribution supported on the roots of \emph{Kravchuk polynomials}.  This implies that the maximum value that $s(x)$ can take on is at least the largest of these roots. This turns out to be exactly the semicircle law, $\on{SCL}_{1/2}(\mu)-o(1)$.

While the proof in \cref{sec:moments-proof} is very concise, we did not see a direct way to improve it to obtain our main results.
In \cref{sec:discrepancy}, we give a second, discrepancy-based re-proof of \cref{prop:dqi-on-opi}, which is more similar to the analysis of DQI in~\cite{dqi}.  While this second re-proof is less concise, understanding why it, our first re-proof, and the original DQI analysis all get stuck at the same place will motivate the main analytical framework that is the backbone of our improved results.  The idea of our second re-proof is as follows.

Define the \emph{$k$-wise discrepancy} for a solution $x$ by
\begin{equation}
    q_k(x):= \sum_{S\in \binom{[m]}{k}} \prod_{i\in S}f_i(\langle b_i, x\rangle) \qquad\text{where}\qquad f_i(a) \coloneqq \begin{cases}
1 &\text{if }a\in S_i\\
-1 &\text{otherwise}
\end{cases}.
\end{equation}
For $\sigma \ll \ell$, DQI samples solution $x$ to Max-LINSAT with probability
\begin{equation} 
\label{eq:dqi-Pu-1}
\mb{P}_u(x) \propto \left(\sum_{k=0}^\ell u_k q_k(x)\right)^2\quad\text{where}\quad u_k\coloneqq\binom{m}{k}^{-1/2}\mbf{1}\{\ell - \sigma\le k\le \ell\}.
\end{equation}
To lower bound $\mb{E}_{\mb{P}_u}[s(x)]$ upon expanding $\mb{P}_u(x)$, we need two key properties of $q_k$.
\begin{enumerate}
 \item Let $C^\perp_t\coloneqq \{y\in \mb{F}_p^{m}:|y|=t, B^\top y=0\}$. With expectation taken uniformly over $\mb{F}_p^n$,
    \begin{equation}\label{eq:prop-1}
        \mb{E}[q_t(x)]=\sum_{y\in C^\perp_t} \prod_{i:y_i\ne 0} \wh{f}_i(y_i).
    \end{equation}
    In particular, $\mb{E}[q_0(x)]=1$ as $C^\perp_0=\{0\}$, and for $0<t<d^\perp=n+1$, $C^\perp_t=\emptyset$, so $\mb{E}[q_t(x)]=0$.
\item Let $N(k_1, \ldots, k_r;t)$ be the number of subsets $T_1, \ldots, T_r$ of $[m]$ where $|T_i|=k_i$ and the set of elements in an odd number of $T_i$'s is exactly $[t]$. With expectation taken uniformly over $\mb{F}_p^n$,
\begin{equation}\label{eq:prop-2}
\mb{E}\left[q_{k_1}(x) \cdot \cdots \cdot q_{k_r}(x)\right] = \sum_{t=0}^{m} N(k_1, \ldots, k_r;t) \mb{E}[q_t(x)].
\end{equation}
\end{enumerate}
These properties are listed as \cref{lem:balanced-fourier,lem:key-fourier}, and \cref{prop:bal-prod,prop:gen-prod,prop:q-triple-rec}, respectively. Using these, we arrive at our main expansion of expected satisfaction ratio: if $2\ell+1<d^\perp$, then
\begin{equation}\label{eq:naive-expansion}
\mb{E}_{x\sim \mb{P}_u} [s(x)]  = \frac{1}{2}+\frac{1}{2m}\frac{\sum_{k, k'=0}^\ell u_ku_{k'}\mb{E}[q_k(x)q_{k'}(x)q_1(x)]}{\sum_{k, k'=0}^\ell u_ku_{k'}\mb{E}[q_k(x)q_{k'}(x)]}
 = \frac{1}{2}+\frac{1}{2m}\frac{\sum_{k, k'=0}^\ell u_ku_{k'}N(k,k',1; 0)}{\sum_{k, k'=0}^\ell u_ku_{k'}N(k,k'; 0)}.
\end{equation}
Plugging in the definitions of $N$ and $u$ gives $\mb{E}[s(x)]\ge \on{SCL}_{1/2}(\ell/m)-o(1)$, recovering the semicircle law. These proofs of \cref{prop:dqi-on-opi} lead to the following natural question:
\begin{question}
Why do the three proofs (DQI, moments-based re-proof, discrepancy-based re-proof) all get stuck at $\on{SCL}_{1/2}(\mu)$?
\end{question}
Restricting \cref{eq:prop-1} to the range $0\le t<d^\perp$ is equivalent to observing $m\cdot s(x)$ has matching moments as $\on{Bin}(m, 1/2)$ of order up to $d^\perp-1=n$, up to a change of basis given by $q_k$ and $u$. In particular, from the moments problem we know there is a distribution $Z$ with such moments whose maximum is exactly the semicircle law, but none of the three approaches can differentiate distributions of $Z$ and $m\cdot s(x)$. This barrier is further discussed in \cref{sec:motivation}.

This motivates what a solution $x$ must look like to beat the semicircle law: we must control \cref{eq:prop-1} for $t\ge d^\perp$, where $\mb{E}[q_t(x)]$ may be nonzero. In \cref{sec:improvements}, we use $\mb{P}_u$ with cutoff $\ell=\lfloor(\mu+\delta)m\rfloor$ for some $\delta>0$, so $2\ell+1\ge d^\perp$ for sufficiently large $m$. Then, in \cref{sec:expansion}, we have analogous to \cref{eq:naive-expansion} that
\begin{equation}\label{eq:better-expansion}
\mb{E}_{\mb{P}_u}[s(x)]
=\frac{1}{2}+\frac{1}{2m}\frac{\sum_{k, k'=0}^\ell u_ku_{k'}N(k,k',1; 0)+\sum_{t=d^\perp}^m \mb{E}[q_t(x)]\sum_{k, k'=0}^\ell u_ku_{k'}N(k,k',1; t)}{\sum_{k, k'=0}^\ell u_ku_{k'}N(k,k'; 0)+\sum_{t=d^\perp}^m \mb{E}[q_t(x)]\sum_{k, k'=0}^\ell u_ku_{k'}N(k,k'; t)}.
\end{equation}
where the first terms of the numerator and denominator correspond to $t=0$ and there are correction terms for $t\ge d^\perp$. If we ignore the correction terms, we see the improvement from $\on{SCL}_{1/2}(\mu)$:
\begin{equation}
\mb{E}_{\mb{P}_u}[s(x)] \approx \on{SCL}_{1/2}(\mu+\delta) > \on{SCL}_{1/2}(\mu).
\end{equation}
It suffices to show that the correction terms in \cref{eq:better-expansion} are exponentially small. In \cref{sec:expansion}, we bound the coefficients in the balanced case by
\begin{equation}\label{eq:E-perfect}
\sum_{k, k'=0}^\ell u_ku_{k'}N(k,k'; t)\le e^{mE(\mu, \delta)+o(m)},
\end{equation}
and similarly for $N(k, k', 1;t)$ term in the numerator.  This leaves us with having to control the terms $\mathbb{E}[q_t(x)]$.

\subsubsection{Connection to leakage-resilient secret sharing}\label{sec:LRS_connection}
To control $\mb{E}[q_t(x)]$, we observe in \cref{sec:step-3} a surprising connection to \emph{local leakage resilience} of Shamir secret sharing. By \cref{eq:prop-1}, we rewrite \begin{equation}\label{eq:balanced-llr}
\left|\mb{E}[q_t(x)]\right|=  2^m \left|\sum_{y\in C^\perp_t} \prod_{i=1}^{m}\wh{\mbf{1}_{S_i}}(y_i)\right|.
\end{equation}
It turns out that this is exactly the sort of expression that has been controlled in the literature on local leakage resilience in secret sharing.
While the literature typically considers \emph{Shamir sharing}~\cite{shamir}---which corresponds to OPI---here we discuss the more general \emph{Massey secret sharing schemes}~\cite{massey01}---which correspond to MDS Max-LINSAT---to better match the notation in the rest of this section.

To share a secret $s$ among $m$ parties, Massey's scheme (over $\mathbb{F}_p$, with respect to an MDS code $\mathcal{C}\subseteq\mathbb{F}_p^{m+1}$ of dimension $n$) chooses a uniformly random codeword $c\in\mathcal{C}$ with $c_0=s$, and gives the $i$'th symbol $c_i$ to party $i$.  Shamir's scheme is the special case when $\mathcal{C}$ is a Reed--Solomon code.

By the MDS property, any $n$ parties can recover the secret, while any $n-1$ learn nothing; this is the desired guarantee in (threshold) secret sharing.  But what happens if more than $n$ parties leak a single bit?  
Formally, we say that a scheme is \emph{one-bit local-leakage resilient} if the following holds.  Suppose that each party can leak a single bit, either $\pm 1$; this results in a \emph{leakage transcript} $L(s) \in \{\pm 1\}^m$, which is a random variable, whose distribution depends on $s$.  Then, for any distinct secrets $s \neq s' \in \mathbb{F}_p$, the total variation distance $d_{\mathrm{TV}}(L(s), L(s'))$ should be small, at most $2^{-\Omega(m)}$.

Over extension fields, it turns out that Shamir's scheme is \emph{not} one-bit leakage resilient~\cite{gw,tyb18}.  However, over prime fields, it turns out that it is, at least for high enough rates $n/m$.  
In more detail, let $L$ denote the leakage transcript under a sharing of a uniformly random secret, and let $L'$ denote the corresponding transcript in the case where all the parties have an independent uniformly random share (that have nothing to do with any secret). When the secret is uniform, the share vector is uniform on the code $C$ obtained by puncturing $\mathcal{C}$ at coordinate $0$. The Fourier calculation of \cite[Lemma 4.14]{bdir} gives 
\begin{equation} \label{eq:BDIR} d_{\mathrm{TV}}(L,L') = \frac{1}{2}  \sum_{\ell \in \{\pm 1 \}^{m} } \left| \sum_{y \in C^\perp \setminus \{0\} } \prod_{i=1}^{m} \hat{\mathbf{1}}_{S_i^{(\ell_i)}} (y_i)  \right|, \end{equation}
where $S_i^{(\ell_i)}$ is the set of shares for party $i$ consistent with the leaked bit $\ell_i$.
The expression \cref{eq:BDIR} looks like \cref{eq:balanced-llr}, except that it is summed up over all possible transcripts $\ell \in \{\pm 1\}^m$.  That is, 
the Fourier sum in \cref{eq:balanced-llr} is the single-weight contribution to a \emph{per-transcript} probability discrepancy of this form. Such discrepancies are bounded in \cite{bdir,mpsw21,mnpw22} in the context of local leakage resilience. The fixed-secret calculation is described in \cref{sec:step-3}.

We could use these existing bounds in an off-the-shelf way with the framework described in \cref{sec:overview1}.  Using the best bound on \cref{eq:balanced-llr} from \cite{mnpw22} already yields some improvement on the semicircle law; this off-the-shelf improvement is plotted as the green curve in \cref{fig:balanced}, and is stated as \cref{thm:main-green}.  The resulting bound on the saturation threshold $\mu_1(1/2)$ is exactly the same as the leakage resilience threshold in \cite{mnpw22}.

In this work, we improve on existing techniques to obtain our improved bounds, \cref{thm:main-avg} and \cref{thm:main-best}.  To explain our improvements, we first explain the approach of prior work.

The idea of \cite{bdir,mpsw21,mnpw22} is to bound \cref{eq:balanced-llr} by splitting the coordinates $[m]$ into two groups $L$ and $R$ of size $k\coloneqq \dim(C^\perp)=m-n$ and a remainder set $B$ of size $m-2k$. By the MDS property of $C^\perp$, the coordinates in $L$ (and $R$) for $y\in C^\perp$ is in bijection with $\mb{F}_p^k$.
By Cauchy-Schwarz,
\begin{equation}\label{eq:perfect-bdir}
\begin{aligned}
\left|\sum_{y\in C^\perp_t} \prod_{i=1}^{m}\wh{\mbf{1}_{S_i}}(y_i)\right| & \le \left(\sum_{y\in C^\perp_t} \prod_{i\in L}\left|\wh{\mbf{1}_{S_i}}(y_i)\right|^2\right)^{1/2}\left(\sum_{y\in C^\perp_t} \prod_{i\in R}\left|\wh{\mbf{1}_{S_i}}(y_i)\right|^2\right)^{1/2}\max_{y\in C^\perp_t} \prod_{i\in B}\left|\wh{\mbf{1}_{S_i}}(y_i)\right|
\\  & \le \left(\sum_{y\in C^\perp} \prod_{i\in L}\left|\wh{\mbf{1}_{S_i}}(y_i)\right|^2\right)^{1/2}\left(\sum_{y\in C^\perp} \prod_{i\in R}\left|\wh{\mbf{1}_{S_i}}(y_i)\right|^2\right)^{1/2}\max_{y\in C^\perp_t} \prod_{i\in B}\left|\wh{\mbf{1}_{S_i}}(y_i)\right|
\\ & \le \left(\prod_{i\in L\cup R} \Vert \wh{\mbf{1}_{S_i}}\Vert_2\right)\cdot \max_{y\in C^\perp_t} \prod_{i\in B}\left|\wh{\mbf{1}_{S_i}}(y_i)\right|.
\end{aligned}
\end{equation}
Now, the key observation is that  $\wh{\mbf{1}_{S_i}}(0)=1/2$ and $\left|\wh{\mbf{1}_{S_i}}(y_i)\right|\le 1/\pi+O(1/p)$ if $y_i\ne 0$. Therefore, we gain for every coordinate in $I\coloneqq \left\{ i\in B:y_i\ne 0\right\}$. Since $y$ has Hamming weight $t\ge d^\perp$, then $|I|\ge t-|L|-|R|$. This leads to the bound given by \cite{mnpw22}. 

To further improve on this bound and prove our main theorems in the balanced case, we split the coordinates into $[m]=L\cup R\cup B$ in $J$ many ways (indexed by $j$), so that each $y$ has larger support $I$ in at least one $B_j$ than the union bound $t-|L|-|R|$ above. We split the sum over $y$ in \cref{eq:perfect-bdir} into $J$ many terms corresponding to the splits, and bound in \cref{thm:step-3} similar to \cref{eq:perfect-bdir}.
\begin{itemize}
    \item If $J=m^{O(1)}$, we can guarantee in \cref{lem:avg-buckets} that $|I|\ge t|B|/m$. Then, we bound \cref{eq:balanced-llr} by $e^{mF(\mu)+o(m)}$. Combined with \cref{eq:E-perfect}, we obtain \cref{thm:main-avg}.
    \item If we tolerate $J$ exponential in $m$ to guarantee $|I|\ge \lambda m$ and optimize this trade-off, then we bound \cref{eq:balanced-llr} by $e^{mG(\mu,\lambda)+o(m)}$ in \cref{lem:best-buckets}. Combined with \cref{eq:E-perfect}, we obtain \cref{thm:main-best}.
\end{itemize}
It is natural to ask whether these improvements can feed back into the leakage-resilient secret sharing literature.  The answer is yes and no: The answer is yes because our bounds do imply improved bounds on the Fourier proxy \cref{eq:BDIR}, which does lead to improvements on the results of \cite{bdir,mpsw21,mnpw22}. 
However, the answer is no because subsequent work~\cite{kk23,k24,n24} has shown how to obtain better bounds on leakage resilience without going through the Fourier proxy \cref{eq:BDIR}, and our techniques do not beat those improvements.  We discuss this more in \cref{sec:llr}. 

It is also natural to ask whether the fact that RS codes are \emph{not} leakage-resilient over extension fields implies negative results for \cref{q:main} over extension fields.  While we can show some weak negative results (see the discussion after \cref{question:perfect-sol} in \cref{sec:discrepancy}), we leave this as an interesting future direction.

To summarize, by exploiting the connection to local leakage resilience, and by further improving techniques from that literature, we show the terms in \cref{eq:better-expansion} with $t\ge d^\perp$ are exponentially small, thereby improving on the semicircle law. 

\subsubsection{Dealing with the details}
Finally, we remark on some subtleties we omitted in the case of a general $\rho$ (the setting of \cref{thm:main-biased}) and the balanced case where $\rho \sim 1/2$ but is not exactly $1/2$. These are spelled out in \cref{sec:framework,sec:expansion}.
The key issue is that the $\pm 1$-valued $f_i$ no longer have $\wh{f}_i(0)=0$, so \cref{eq:prop-1} does not hold. Instead, we work with a linear transformation $g_i$ and define discrepancy $q_k$ based on $g_i$.
   Then, $g_i$ is not $\pm 1$-valued, so \cref{eq:prop-2} no longer holds. Instead, we obtain analogous control to bound $\mb{E}[q_k(x)q_{k'}(x)]$ and $\mb{E}[q_k(x)q_{k'}(x)q_1(x)]$ based on the same idea of considering the symmetric difference, and replacing the quantity $N$ with a weighted version $N_\rho$. 
  
  For the balanced case where $\rho\sim 1/2$, $N_\rho \approx N$ up to lower order terms, and we recover the above analysis up to lower order terms. This is given in \cref{sec:balanced-proofs}.
For the general case $\rho$, the function $E_\rho$ involves itself an optimization, shown in \cref{thm:main-biased}. We adapt the Fourier analysis and the polynomial $J$ control idea adapted to the general case of $|S_i|\sim \rho p$, and obtain \cref{thm:main-biased} in analogy with \cref{thm:main-avg}. This is given in \cref{sec:general-proofs}.

\subsection{Related Work}
\label{sec:related-works}
In this section, we put our results and techniques in the context of related work.  
\subsubsection{DQI and algorithms based on Regev's reduction}\label{sec:dqi-related}
OPI has been studied recently as a target of quantum algorithms, including DQI~\cite{dqi} and related algorithms inspired by Regev's reduction~\cite{regevReduction}, including~\cite{chailloux1,chailloux2,khattar2025verifiable,rosmanis2026nearly}.

First, we mention the relationship between our re-proofs of the semicircle law and existing work.  As discussed above, the analysis of DQI, as well as our re-proofs of \cref{prop:dqi-on-opi} that only use the moments of $s(x)$ of order up to $n$, all encounter the semicircle law barrier.  This is because the proofs boil down to constructing a polynomial of degree at most $n$ to certify that $\max s(x)$ cannot be too small. Applying this polynomial certificate construction to coding theory dates back to \cite{tie} on the covering radius problem, where the bound is also given by a semicircle law. Indeed, it is observed in \cite{dqi-complexity} that DQI makes this approach algorithmic, that sampling a state of the form \cref{eq:dqi-Pu-1} with the best weights $w$ can be thought of as the construction of such a polynomial certificate. In \cite{dqi-complexity}, further connections of the DQI algorithm and \cite{tie} via the MacWilliams identities are discussed and DQI is shown to be simulated in a low level of the polynomial hierarchy, ruling out certain hardness arguments.

Next, we discuss results beyond the semicircle law.  To the best of our knowledge, there were previously no results, algorithmic or otherwise, that yielded an improvement on \emph{worst-case} instances.  However, there are average-case results, both for the case of uniformly random lists~\cite{dqi} and for a more structured random ensemble~\cite{chailloux1,chailloux2} described below.

The idea of choosing $\ell > d^\perp/2$ to gain further performance improvement is already present in~\cite[Section 10]{dqi}. Motivated by empirical performance of belief propagation on LDPC codes, the authors observe over $\mb{F}_2$ that DQI has strong performance guarantees in the \emph{average} case where each singleton input list $S_i$ is a uniform random bit, provided the decoder of $C^\perp$ works beyond $d^\perp/2$ with small failure probabilities. In that analysis, averaging over the binary labels cancels the terms from nonzero dual codewords. In \cref{sec:improvements} we bound the corresponding terms for each fixed collection of lists using the MDS property and prime-field Fourier bounds.

The idea of choosing $\ell > d^\perp/2$ is also present in \cite{chailloux1, chailloux2}.  DQI can be seen as an instantiation of Regev's reduction~\cite{regevReduction}.  The works \cite{chailloux1,chailloux2} combine similar ideas of Regev's reduction with soft-decoders to give quantum algorithms that improve on DQI on average for the following structured ensemble: Lists $\tilde{S}_i \subseteq \mathbb{F}_p$ are chosen in a worst-case way, and the final input lists $S_i$ are given by $S_i = \tilde{S}_i + e_i$, where $e_i \in \mb{F}_p$ are independent uniformly random shifts. This is essentially a hybrid between the worst-case and the (uniform) average-case considered in \cite{dqi}. In this hybrid case and assuming $\rho =1/2$, the strongest algorithm \cite[Algorithm 4]{chailloux2} uses the Koetter-Vardy soft decoder and finds a solution with satisfaction ratio at least $1-o(1)$ whenever the rate is at least $3/4$. As noted in \cref{rem:chailloux}, our analogous threshold in the worst case is $2\mu_1(1/2)\le 0.7496 < 3/4$ from \cref{cor:feature-cor}, giving a small quantitative improvement over the result of \cite{chailloux2} in this parameter regime, even ignoring the difference between worst-case and the hybrid model.  On the other hand, our result is not algorithmic, while that of \cite{chailloux2} is; moreover in the low-rate regime the results of \cite{chailloux2} are quantitatively stronger than our results.  

However, in some sense the results described above are orthogonal to our work.  We focus on existential guarantees for worst-case input lists $S_i$.  In both the (uniformly) average case and the hybrid model discussed above, a simple first-moment argument shows the existence, in expectation over the inputs, of a perfect solution that satisfies every constraint. From this perspective, the difficulty we tackle is to extend the argument to the worst-case input lists $S_i$. 

Finally, we note that while it is possible to algorithmically find worst-case solutions to OPI up to the semi-circle law (and works subsequent to ours~\cite{HY26,Jo26} have done the same even beyond the semi-circle law), such worst-case algorithmic results seem unlikely for more general (unstructured) cases of max-LINSAT.  For example, the recent work \cite{KSE26} has shown that it is NP-hard to approximate general max-LINSAT instances beyond the random-assignment threshold.

\subsubsection{Local Leakage Resilience}
\label{sec:llr}
As mentioned above, our key technical insights are inspired by the literature on local leakage resilience in secret sharing.  
We briefly summarize the relevant literature.  As above, we use the language of Massey's secret sharing scheme (rather than Shamir) to better match our notation.

We have already defined one-bit local leakage-resilience above.
The main conjecture in this setting, first stated for Shamir's scheme, is the following.
\begin{conjecture}[{\cite[Conjecture 1.3]{bdir}}]
\label{conj:llr-sss}
The Massey scheme defined with an MDS code $C$ over a prime field $\mb{F}_p$ is one-bit local leakage resilient for any positive rate.
\end{conjecture}
For lower bounds, \cite{ns} shows the statement of \cref{conj:llr-sss} is false without the positive rate assumption, namely Shamir secret sharing is not leakage resilient with rate $O(1/\log n)$. 

For upper bounds, \cite{bdir,mpsw21,mnpw22} proved \cref{conj:llr-sss} holds for rate at least $0.78$, using the approach outlined in \cref{sec:overview}. In more detail, they first bound \cref{eq:balanced-llr}, the term that we also want to bound in the balanced case.  Then, by pairing up the transcripts $\ell$ in the outer sum of \cref{eq:BDIR}, the work \cite{mnpw22} argues that the right side is maximized when the leakage functions are balanced (that is, when the corresponding sets $S_i$ have size $\sim p/2$).  Putting these together proves \cref{conj:llr-sss} for rates at least $0.78$.

We make a few remarks about the relationship between existing work on leakage-resilient secret sharing and our work on OPI.  

First, as discussed in \cref{sec:overview}, existing bounds on \cref{eq:balanced-llr} give a weakened version of our main results in the balanced case (see \cref{thm:main-green}).  However, while it is sufficient for \cite{mnpw22} to work in the balanced case, we would like a proof for general $\rho$, which adds another layer of complexity to our improvements.

Second, as mentioned in \cref{sec:overview}, one might hope that combining the transcript pairing idea with our improved control of \cref{eq:balanced-llr} (\cref{lem:best-buckets}), that we could give improved results for secret sharing.  What we get is the following proposition.
\begin{proposition}\label{prop:our-llr}
    The Massey scheme on MDS code $C$ over a prime field $\mb{F}_p$ is one-bit local leakage resilient for rate at least $0.7496$.
\end{proposition}
\cref{prop:our-llr} gives the state-of-the-art bound on the Fourier proxy \cref{eq:BDIR}, and indeed improves on the $0.78$ bound of \cite{mnpw22}.  However, it does \emph{not} beat the state-of-the-art bounds on \cref{conj:llr-sss}, which do not use the Fourier proxy \cref{eq:BDIR}.  
The motivation to break with the Fourier proxy is the observation in \cite{n24} that any approach bounding \cref{eq:BDIR} hits a barrier at rate $1/2$ via an explicit quadratic residue example. Thus, our results on the saturation threshold in the balanced case, as well as \cref{prop:our-llr}, suffer from the same barrier. 

To go beyond this barrier, \cite{n24} employs higher order Fourier analysis using Gowers $U^s$ norms to show \cref{conj:llr-sss} holds with high probability over random leakages. This is analogous to the analyses of average-case $S_i$ (and the more structured hybrid ensemble) discussed in \cref{sec:dqi-related}.  In a breakthrough, \cite{kk23,k24} apply Fourier analysis over the transcripts $\ell\in\{-1, 1\}^{m}$ to instead bound total sum via an $L^4$-Fourier proxy and showed \cref{conj:llr-sss} holds for rate at least $0.668$. 
We make two remarks regarding this improvement.  The first is that, unfortunately, these methods are not applicable to the OPI problem, since we really do want to bound \cref{eq:balanced-llr}, rather than using it as a proxy to bound something else.
The second is that, under the Fourier proxy framework (bounding \cref{eq:BDIR}), it is clear by pairing transcripts that the balanced leakages give the worst case \cite{mnpw22}. However, \cite{kk23,k24} improves exactly by replacing the transcript pairing with some Boolean Fourier analysis, so the balanced leakage $S_i$ case may no longer be the worst case. For example, they consider cases where $S_i$ are not all of the same size.

\subsubsection{Other work on OPI}\label{sec:OPI}
As mentioned at the beginning of the paper, OPI and related problems arise in many domains.  However, to the best of our knowledge, none of these (except the work on DQI and related algorithms already discussed in \cref{sec:dqi-related}) are concerned with \cref{q:main}.  We briefly discuss two related areas below: list-recovery of Reed-Solomon codes, and OPI as a cryptographic assumption.

In the coding theory literature, OPI arises in the context of \emph{list-recovery} of Reed-Solomon (RS) codes.  We say that a code $C \subseteq \mathbb{F}^m$ is \emph{$(\alpha, \ell, L)$-list-recoverable} if, for any $S_1, \ldots, S_m \subseteq \mathbb{F}$ each of size at most $\ell$, there are at most $L$ codewords $c \in C$ so that $c_i \in S_i$ for at least a $\alpha$ fraction of $i \in [m]$.  The algorithmic problem is to efficiently output all at-most-$L$ such codewords.  In the language of OPI, when the code $C$ is a RS code, the list-recovery problem is to return all polynomials $Q$ of degree less than $n$ so that the satisfaction ratio $s(Q)$ is at least $\alpha$. 

List-recovery of RS codes (and related codes) has been extensively studied; for example, the celebrated Guruswami-Sudan algorithm~\cite{gs} list-recovers RS codes up to a limit on $\alpha$ known as the \emph{Johnson bound}.  The work \cite{GR05} showed that, for full-length RS codes over extension fields, the Johnson bound is the correct limit, in the sense that beyond that there will be an exponential number of codewords that agree a lot with the input lists $S_i$ in the worst case.  More recent work has established that RS codes with \emph{random} evaluation points are list-recoverable all the way up to the information-theoretic limit on $\alpha$, beyond the Johnson bound~\cite{BCDZ25}.  

If $L$ is small, an efficient list-recovery algorithm for RS codes solves OPI: simply run the algorithm with the input lists $S_1, \ldots, S_m$, and iterate through the list of at most $L$ options to find the best solution.  However, the OPI problem is most interesting (from the perspective of potential quantum advantage) when the list size $L$ is very large.  In these parameter regimes, list-recovery algorithms cannot be efficient (as it would take too long to output the list).  Moreover, all existing algorithms we are aware of do not ``fail gracefully'' in this parameter regime to return a decent solution to OPI.  For example, the Guruswami-Sudan algorithm, which works by interpolating a polynomial to ``explain'' potential solutions, fails to interpolate an appropriate polynomial beyond the Johnson bound.

OPI has been studied in the cryptography literature under the name \emph{noise polynomial interpolation} (or \emph{noisy polynomial reconstruction}).  For example, \cite{np99} proposed a version of OPI as a hardness assumption; later, \cite{bn00} broke this assumption via a lattice attack.  This setting differs from ours for two main reasons.  The first is that the input lists are taken to be random, not worst-case.  The second is the the parameter regime is the ``planted solution'' regime: the parameters are such that randomly generated lists will typically have no good solutions, and the computational challenge is to distinguish a random instance from one where a good solution is planted.  In contrast, in the OPI problem we are more generally interested in finding a good solution out of potentially many good solutions.

\subsubsection{Subsequent Work}
\label{sec:subsequentwork}
Since this work was originally posted on arXiv, two follow-up works~\cite{HY26,Jo26} have improved our results and made them algorithmic.

First, we mention the work of Horinaga and Yamakawa~\cite{HY26}, which has given a quantum algorithm to find solutions with satisfaction beyond the semi-circle law, and also improved our existential results.  In more detail, they give a quantum algorithm that can find perfect solutions to OPI for $R > 0.75$, matching the average-case result of \cite{chailloux2}. For $R\le 0.75$, the satisfaction fraction of their algorithmic solutions are similar in spirit but not comparable with our existential solutions. 
They also improve our existential results: In particular, show that perfect solutions exist for codes of rate $R > 0.7158$, improving over our bound of $0.7495$.  

Next, we mention the work of Jo~\cite{Jo26}, which was concurrent with \cite{HY26}.  This work gives an algorithmic improvement over the semicircle law whenever the underlying Reed--Solomon code has rate $R \geq 0.6225$, and finds an asymptotically perfect OPI solution whenever $R \geq 0.75$, again matching the average-case result of \cite{chailloux2}.  The work \cite{Jo26} explicitly builds on our techniques, showing how to sample efficiently from the distribution $P_u$ that we construct as part of our approach.  

\subsection{Discussion and Open Problems}\label{sec:discussion}
We discuss a number of open problems and future directions.
The immediate open question is to further improve the guarantees on the satisfaction ratio attainable in the worst case. We highlight the saturation threshold $\mu_1(\rho)$ for its connection with the local leakage resilience threshold in \cref{conj:llr-sss}.
\begin{question}
What is the true value of $\mu_1(\rho)$ and in particular $\mu_1(1/2)$?
\end{question}
We know $2\mu_1(1/2)\le 0.7496$. By combining our techniques in \cref{thm:main-biased} for general $\rho$ with our techniques from \cref{thm:main-best} with the best Fourier control for $\rho \sim 1/2$, we believe we could improve on \cref{thm:main-biased} in the general $\rho$ case.  We suspect that further improvements are possible. However, our approach meets 
a barrier at rate $1/2$, as discussed after \cref{prop:our-llr}.
For the general $\rho$ case, an immediate question is to remove the defect for $\rho\ge 1/2$ that causes the non-monotonicity in \cref{fig:phase}, since we know $\mu_1(\rho)$ to be decreasing in $\rho$. This is further discussed in \cref{rem:tildeN}. 


\vspace{.5cm}
The second question focuses on the cases above the $\mu_1(\rho)$ threshold, where the asymptotically perfect solutions exist when $\mu\ge \mu_1(\rho)$, i.e. $x$ where $s(x)\ge 1-o(1)$. We wish to remove the $o(1)$.
\begin{question}\label{question:perfect-sol}
For the Max-LINSAT problem with input lists $|S_i|=\rho p$ over MDS codes of rate $2\mu$, if $\mu\ge \mu_1(\rho)$, is there a solution $x$ that satisfies \emph{every} constraint?
\end{question}
When the field size is a prime power but not prime, we have negative results via Reed-Solomon code repair literature. 
See \cref{rem:p_prime} for a more precise discussion of where our approach (and that of leakage resilience literature) breaks down for extension fields.  We note that in the leakage resilience literature, there is a provable difference between extension fields and prime fields, and so that may be the case here as well.

We present some negative evidence for \cref{question:perfect-sol} over $\mathbb{F}_q$.  Recall a result from~\cite{gw} on the repair of Reed-Solomon codes.  In that work, the authors prove that for the full-length Reed-Solomon code $C$ of rate at most $1/2$ over $\mb{F}_{q}$, where $q=2^t$ is a power of $2$, we can construct subspaces $S_i\subset \mb{F}_q$ of co-dimension one for each $i\in\mb{F}_q^\times$, such that if a codeword $c\in C$ satisfies $c_i\in S_i$ for every $i$, then the value of $c_0$ is determined. By choosing $S_0$ to not include this value, we answer \cref{question:perfect-sol} in the negative for these choices of $C$.   Answering \cref{question:perfect-sol} in the affirmative for prime fields would demonstrate an interesting gap between prime and non-prime fields. 

\vspace{.5cm}
Our last question focuses on the connection with local leakage resilience that we have exploited.
\begin{question}\label{question:apply-to-llr}
Can our techniques be used to improve on the threshold rate of local leakage resilience of Massey secret sharing schemes, e.g. \cref{conj:llr-sss}?
\end{question}
Recall that \cref{prop:our-llr} partially answers this problem in the positive, if we use the Fourier proxy in \cite{bdir,mpsw21,mnpw22}. However, our methods are not immediately compatible with the state-of-the-art approaches to \cref{conj:llr-sss}, as previously discussed in \cref{sec:related-works}.

\subsection*{Acknowledgements} 
We thank Noah Shutty for helpful conversations, and we thank Ankur Moitra for helpful conversations and for pointing out Tiet{\"a}v{\"a}inen's bound.  We also thank anonymous reviewers for helpful feedback. YS is funded by the NSF Graduate Research Fellowship and the Stanford Graduate Fellowship. MW is partially funded by NSF grants CCF-2231157 and CNS-2321489.

\section{Preliminaries}
\label{sec:prelim}

\subsection{Coding Theory}
\label{sec:coding}
We start with some background on coding theory,  
though with some non-standard parameter naming conventions to match the Max-LINSAT problem and \cite{dqi}.

A \emph{linear code} $C$ of \emph{length} $m$ and \emph{dimension} $\dim(C) = n$ over $\mathbb{F}_q$ is a linear subspace
$C\subseteq \mathbb{F}_q^m$ and dimension $n$. 
For $c\in C$, the \emph{Hamming weight} $|c|$ is the number of
nonzero coordinates of $c$, and the \emph{minimum distance} $d$ of $C$ is the minimum Hamming weight of a non-zero codeword $c\in C$.
Equivalently, $d$ is the minimum Hamming distance between two distinct codewords. We say that $C$ is an $[m,n,d]_q$ code.
We say $B\in \mathbb{F}_q^{m\times n}$ is a \emph{generator matrix} for $C$ if
$C=\{Bx:x\in\mathbb{F}_q^n\}$. Define
\[
C^\perp \coloneqq \{y\in \mathbb{F}_q^m:\langle x,y\rangle =0\,\forall\, x\in C\}
= \{y\in \mathbb{F}_q^m: B^\top y = 0\},
\]
to be the \emph{dual code} of $C$; it has dimension $m-n$. Let
$
C_t\coloneqq\{c\in C:|c|=t\}$ for integer $t\ge 0$.

The \emph{Singleton bound} for a $[m,n,d]_q$ linear code $C$ states that
$
d\le m-n+1
$. Codes meeting this bound are \emph{maximum distance separable} (MDS). We use two standard facts: $C$ is MDS if and only if every set of $n$ rows of a generator matrix $B$ is linearly independent; if $C$ is MDS, then so is $C^\perp$.

Fix $m$ distinct evaluation points $a_1,\ldots,a_m\in \mathbb{F}_q$ and an integer $1\le n\le m$.
The \emph{Reed-Solomon code} on the evaluation set
$\mathbf{a}=(a_1,\ldots,a_m)$ is
\[
\on{RS}_q(\mathbf{a}, n)\coloneqq \{(f(a_1),\ldots,f(a_m)) : f\in \mathbb{F}_q[x],\ \deg f<n\}\subseteq \mathbb{F}_q^m.
\]
Equivalently, $\on{RS}_q(\mathbf{a}, n)$ is generated by the $m\times n$ Vandermonde matrix
\[
B=
\begin{pmatrix}
1 & a_1 & a_1^2 & \cdots & a_1^{n-1}\\
1 & a_2 & a_2^2 & \cdots & a_2^{n-1}\\
\vdots & \vdots & \vdots & \ddots & \vdots\\
1 & a_m & a_m^2 & \cdots & a_m^{n-1}
\end{pmatrix},
\]
It is standard that $\on{RS}_q(\mathbf{a}, n)$ is an $[m,n,m-n+1]_q$ MDS code. Therefore, $\on{RS}_q(\mathbf{a}, n)^\perp$ is also MDS.

We use the following lemma that shows the alphabet size $q$ must grow linearly as $m, n\to\infty$ with fixed rate $2\mu \in (0, 1)$, as in the setting of our asymptotic regime.\footnote{
This is a weaker statement of the \emph{MDS conjecture}, which states that every nontrivial linear MDS code
over $\mathbb{F}_q$ should have length $m\le q+1$, except the exceptional families with even-$q$ and $n\in\{3,q-1\}$ where $m\le q+2$. Equivalently, (full-length) Reed-Solomon
codes should be extremal.
For prime fields this conjecture was proved by Ball \cite{ball}; over general extension
fields many parameter ranges are known, but the full conjecture is still open; see, for
instance, the survey \cite{ball-survey}.  However, the weaker statement is enough for our purposes.}

\begin{lemma}[{\cite[Lemma 1.2]{ball}}]\label{lem:mds-bound}
Let $C$ be a nontrivial linear $[m,n,d]_q$ MDS code, i.e. assume $2\le n\le m-2$. Then, $q\ge \max\{m-n+1,\;n-1\}$.
Hence, $q\ge m/2$ up to an additive constant.
\end{lemma}

\subsection{The Moments Problem and Kravchuk Polynomials}\label{sec:moments}
The \emph{(discrete) moment problem} asks whether a finite list of numbers $(m_0, \ldots, m_n)$ can be realized as the moments of a measure $\sigma$ supported on a (finite) set, and how to reconstruct such a measure from those moments. For a standard reference on the moments problem, see for example \cite{akh}.

Given the moment sequence, we can define a bilinear form on polynomials by
$
\langle x^a,x^b\rangle \coloneqq m_{a+b}$, 
and extends bilinearly. When the moments come from a probability measure $\sigma$,
this is simply
\[
\langle P,Q\rangle = \mathbb{E}_{X\sim \sigma}[P(X)Q(X)].
\]
Applying Gram--Schmidt to $1,x,x^2,\ldots$ produces orthogonal polynomials
$p_0,p_1,\ldots$. Crucially, the polynomial $p_j$ depends only on the moments up to
order $2j-1$; equivalently, the family $p_0,\ldots,p_\ell$ is determined by the truncated
moment sequence $m_0,\ldots,m_{2\ell-1}$ for each $\ell$. 
For a discrete distribution $\sigma$ on $\{0,\ldots,m\}$, the relevant inner product is
\begin{equation}\label{eq:orth-rel}
\langle P,Q\rangle_\sigma=\sum_{x=0}^m \sigma(x)P(x)Q(x).
\end{equation}
When $\sigma=\mathrm{Bin}(m,1/2)$, the orthogonal polynomial family is the \emph{Kravchuk}
family. Since they will be key objects for us, we define them explicitly and suggestively to indicate the association.
\begin{definition}
\label{def:kravchuk}
The \emph{degree-$\ell$ Kravchuk polynomial} $K_\ell(x)=K_\ell(x; m, 1/2)$ associated with $\on{Bin}(m, 1/2)$ distribution is defined via its generating function: the Kravchuk polynomials $K_\ell(x)$ satisfy that
\begin{equation}
\label{eq:kravchuk-gen-func}
G(x, z)\coloneqq (1+z)^{m-x}(1-z)^{x} = \sum_{\ell =0}^{m}K_\ell(x)z^\ell.
\end{equation}
\end{definition}
The closed form of $K_\ell$ and its scaled monic version $k_\ell$ is given by
\begin{equation}
\label{eq:kravchuk-def}
K_\ell(x) \coloneqq \sum_{j=0}^\ell (-1)^j \binom{x}{j}\binom{m-x}{\ell-j} \qquad\text{and}\qquad k_\ell(x) \coloneqq \ell!(-2)^{-\ell}K_\ell(x).
\end{equation}
They satisfy orthogonality with respect to the inner product \cref{eq:orth-rel} with $\sigma =\on{Bin}(m, 1/2)$, i.e.
\[
\sum_{x=0}^m \binom{m}{x}K_r(x)K_s(x)
=
2^m\binom{m}{r}\mathbf{1}\{r=s\}.
\]
We use an asymptotic computation of the roots of $K_\ell$. Let $Z_{\max}(Q)$ be the largest root of $Q$.
\begin{fact}[{\cite[Theorem 1.10]{kva}}]
\label{fact:Nki-kravchuk}
Suppose $\ell \sim \mu m$ for fixed $\mu\in [0, 1]$ as $m\to\infty$, then
\begin{equation}
\lim_{m\to\infty} \frac{1}{m}Z_{\max} \left(K_{\ell}\left(x; m, \frac{1}{2}\right)\right) = \on{SCL}_{1/2}(\mu).
\end{equation}    
Moreover, the largest root divided by $m$ of the degree $\ell$  Kravchuk family associated to $\on{Bin}(m, \rho)$ is precisely $\on{SCL}_\rho(\mu)-o(1)$ as $m\to\infty$, where $\rho\in (0, 1)$ is fixed.
\end{fact}
Further discussions on orthogonal polynomials and the Kravchuk family can be found in \cite{orth}.
Finally, circling back to the discrete moments problem, we cite the standard result we will use.

\begin{theorem}[Chebyshev-Markov-Stieltjes, {\cite[Theorem 2.5.4]{akh}}]
\label{thm:cms}
Consider the moments problem $(m_k)_{k=0}^{2\ell-1}$, let $p_0, \ldots, p_\ell$ be the first $\ell+1$ orthogonal polynomials with respect to the moments, let
\begin{equation} q(z)\coloneqq \left(\sum_{k=0}^\ell p_k(z)^2\right)^{-1},
\end{equation}
and let $z_{1}<\ldots <z_\ell$ be roots of $p_\ell$. Then, the random variable $Z$ supported on the $\ell$ roots with mass $q(z_k)$ at $z_k$ is the principal representation of the moments, i.e. it satisfies the following:
\begin{enumerate}
    \item $Z$ solves the moments problem, i.e. $\mb{E}[Z^k]=m_k$ for $0\le k\le 2\ell-1$.
    \item Among solutions of the moment problem, the distribution of $Z$ has the minimum support size. That is, $|\on{supp}X|\ge \ell$ for any other solution $X\sim \sigma$ of the moments problem.
    \item For any other solution $X\sim \sigma$ of the moments problem the distributions must interlace. That is, for every $0\le j\le \ell$,
    \begin{equation} \mb{P}(X\le z_j)\le \mb{P}(Z\le z_j)\coloneqq \sum_{k=1}^j q(z_k) \le \mb{P}(X\le z_{j+1}).\end{equation}
    \item If $\sigma$ has support of size larger than $\ell$, then the inequalities above can be made strict.
\end{enumerate}
\end{theorem}

\subsection{Notations and Conventions}
We let $[n]\coloneqq \{1, \ldots, n\}$ with the convention $[0]=\emptyset$.  Let $\binom{S}{k}$ denote the set of $k$-element subsets of $S$. We use ${a \choose b}$ to be the standard binomial coefficient, with the convention that it is zero if $b$ is not an integer. 
Let $\sqcup$ denote disjoint unions, viewed as multisets in the case of repeated elements.
Let $\langle u, v\rangle$ denote the standard dot-product. Let $\mb{E}_{x\in X}$ denote the expectation over $x\sim\on{Unif}(X)$.

We use the following Fourier analysis convention over $\mb{F}_p^n$: let $e_p(t)\coloneqq\exp(2\pi i t/p)$, and 
\begin{equation}
\begin{aligned}
\wh{f}(y) & \coloneqq\mb{E}_{x\in \mb{F}_p^n} \left[f(x)e_p(\langle x, y\rangle)\right],\\
f(x)&\coloneqq\sum_{y\in \mb{F}_p^n} \wh{f}(y)e_p(-\langle x, y\rangle).
\end{aligned}
\end{equation}

We write $f(x)\propto g(x)$ if there exists constant $C$ such that $f(x)=Cg(x)$ for every $x$. We also adopt standard notation from asymptotic analysis: as $x\to\infty$, we write $f(x)\ll g(x)$ or $f(x)=o(g(x))$ if $f(x)/g(x)\to 0$; $f(x) \lesssim g(x)$ or $f(x)=O(g(x))$ if there exists a finite, positive constant $C$ such that $f(x)\le Cg(x)$ for all sufficiently large $x$; and we write $f(x)\asymp g(x)$ or $f(x)=\Theta(g(x))$ if $f(x)\lesssim g(x)$ and $g(x)\lesssim f(x)$. 
We also write $f(x)\sim g(x)$ if $f(x)/g(x)\to 1$.

In this paper, the asymptotics are always as $m, n\to\infty$ with fixed rate $n/m=2\mu$. By \cref{lem:mds-bound}, the field size (which we call $q$ or $p$, depending on whether it is a general prime power or a prime) also goes to infinity at least linearly in $m$ and $n$; this is because we are interested in MDS codes with block length $m$, which by \cref{lem:mds-bound} only exist over sufficiently large fields.

In this paper, all logarithms are natural logs. Recall from \cref{eq:Hx} the binary cross-entropy function $H:[0, 1]\to \mb{R}$ defined by $H(0)=H(1)=0$ by continuity, and for $x\in (0, 1)$ we let
\[
    H(x)\coloneqq -x\log x-(1-x)\log (1-x).
\]
$H$ is undefined outside of $[0, 1]$. When we maximize or minimize an objective containing $H$ over some parameters, we implicitly do so over the parameter space where all arguments of $H$ are between $0$ and $1$.
We use the following fact about $H$ and Stirling's approximation (see, e.g., \cite[Chapter~1]{debruijn}).
\begin{fact}
\label{fact:stirling}
For $x\in (0, 1)$, $H'(x)=\log(1/x-1)$. For constant $\alpha\in (0, 1)$ as $m\to\infty$
\begin{equation}\label{eq:stirling}
\binom{m}{ \lfloor \alpha m\rfloor} \asymp m^{-1/2}\exp \left(mH(\alpha)\right)\quad\text{where}\quad H(x)\coloneqq -x\log x-(1-x)\log (1-x).
\end{equation}
\end{fact}

\section{Rediscovering the Semicircle Law}\label{sec:re-proof}
In this section, we give several closely related proofs of \cref{prop:dqi-on-opi}
and explain their equivalence. This motivates some definitions useful for later sections that improve \cref{prop:dqi-on-opi}. For ease of exposition, we assume the \emph{perfectly balanced} case where input lists $S_i$ all have size exactly $q/2$. It is possible to give similar re-proofs of \cref{prop:dqi-on-opi} for the case where $|S_i|\sim q/2$ (rather than exactly equal) and more generally when $|S_i|=\rho q$, using machinery developed in \cref{sec:improvements}. We omit such re-proofs, as the point of this section is simplicity and intuition.
\subsection{The Moments Problem Proof}\label{sec:moments-proof}
Using machinery from \cref{sec:moments}, we give a short second proof of \cref{prop:dqi-on-opi}. The ``first'' proof of it is the DQI proof-by-(quantum)-algorithm~\cite{dqi}.

\begin{proof}[Second Proof of \cref{prop:dqi-on-opi}]
Sample $x\in \mb{F}_q^n$ uniformly at random. Then, the collection of $m$ random variables $\{(Bx)_i:i\in [m]\}$ is $n$-wise independent by the MDS property of $C$, so the number $m\cdot s(x)$ of satisfied constraints has the same first $n$ moments as $\on{Bin}(m, 1/2)$.
By \cref{thm:cms}, 
$
\mb{P}\left(m\cdot s(x)\ge Z_{\max}(K_\ell)\right)>0$ where $Z_{\max}(K_\ell)$ is the largest root of the degree $\ell=\lfloor(n+1)/2\rfloor$ Kravchuk polynomial $K_\ell$. Now, \cref{prop:dqi-on-opi} follows \cref{fact:Nki-kravchuk}.
\end{proof}
\begin{remark}\label{rem:momentspf}
We make two remarks.  First, we remark that our application of the MDS property is the generalization of the following fact: for a uniformly random polynomial $Q\in \mb{F}_p[t]$ with $\deg Q\le n$ conditioned on the value of $Q$ at any subset $I\subset \mb{F}_p$ of size at most $n$, $Q(a)$ remains uniformly random for any $a\not\in I$. Indeed, this is the special case when $C$ is the Reed-Solomon code as in OPI.

Second, we remark that the exact same proof using Kravchuk polynomials associated with $\on{Bin}(m, \rho)$ and its largest root work for the general case $|S_i|=\rho q$ as well, and we recover $\on{SCL}_\rho(\mu)$.
\end{remark}

Somewhat shockingly, the above recovers the same semicircle law that DQI achieves. However, in retrospect, this may not be so shocking; in fact, the two proofs are in some sense equivalent.  To explain the coincidence,
 we first define $f_i:\mb{F}_p\to \{-1, 1\} $ for each $i\in [m]$ by
\begin{equation}
\label{eq:fi}
f_i(a) \coloneqq \begin{cases}
1 &\text{if }a\in S_i\\
-1 &\text{otherwise}
\end{cases}.
\end{equation}
We recall from {\cite[Lemma 9.2]{dqi}} how DQI ends up at the semicircle law.
For any $\ell\in\mb{N}$ and weights $w\in \mb{R}_{\ge 0}^{\ell+1}$, assuming efficient syndrome decoding of $C^\perp$ to weight $\ell$, DQI prepares the state
\begin{equation} 
\label{eq:dqi-state}
\ket{P_f} \propto \sum_{k=0}^\ell w_k \ket{P^{(k)}}\quad\text{where}\quad  
\ket{P^{(k)}} \propto \sum_{x\in \mb{F}_p^n}\sum_{S\in \binom{[m]}{k}} \prod_{i\in S}f_i\left(\langle b_i, x\rangle\right) \left|x\right\rangle,
\end{equation}
takes a measurement, and return the resulting solution $x$.
The expected satisfaction ratio of $x$ is
\begin{equation}
\label{eq:dqi-quadratic-form}
\mb{E}s(x) = \frac{1}{2}+\frac{\langle w,  A^{(\ell)}w\rangle }{2m\langle w, w\rangle},
\end{equation}
where $A^{(\ell)}\in \mb{R}^{(\ell+1)\times (\ell+1)}$ is symmetric with $A^{(\ell)}_{k, k-1}=A^{(\ell)}_{k-1, k} =\sqrt{k(m+1-k)}$ for each $k\in [\ell]$  and all other entries zero. The optimal feasible choice of $(\ell, w)$ is given by
$\ell =\lfloor d^\perp/2-1\rfloor$ and $w$ that is the eigenvector associated with the largest eigenvalue $\lambda_{\max}(A^{(\ell)})$. Then, \cref{eq:dqi-quadratic-form} is equal to $\on{SCL}_{1/2}(d^\perp/2m)-o(1)$.

We explain the emergence of the semicircle law by showing that the monic Kravchuk polynomial $k_{\ell+1}$ is the characteristic polynomial of $A^{(\ell)}$ for each $\ell$, up to a constant shift as in \cref{eq:dqi-quadratic-form}. We prove this by showing they obey identical three-term recursions. Then, the number of satisfied constraints $m/2+\lambda_{\max}(A^{(\ell)})/2$ is precisely the largest root of $k_{\ell+1}$, which obeys the semicircle law by \cref{fact:Nki-kravchuk}.

\begin{proposition}\label{prop:3-term-rec}
For every $\ell\in\mb{N}$, let $k_\ell$ be the monic Kravchuk polynomial of degree $\ell$. Then
\begin{equation}
    k_{\ell+1}(x) = \det\left(\left(x-\frac{m}{2}\right)I_{\ell+1}-\frac{1}{2}A^{(\ell)}\right).
\end{equation}
\end{proposition}
\begin{proof}
Let $P_{\ell+1}(x)$ be the polynomial on the right. 
Clearly, $k_{\ell+1}$ and $P_{\ell+1}$ are both monic and have degree $\ell+1$.
Note that $k_0(x)=1=P_0(x)$ and $k_1(x)=x-m/2=P_1(x)$. 
By the recursive formula of the determinant on the right via expanding the last row of the matrix, we obtain that
\begin{equation}\label{eq:three-term-recursion}
    P_{\ell+1}(x) = \left(x-\frac{m}{2}\right)P_{\ell}(x)-\left(-\frac{A^{(\ell)}_{\ell-1, \ell}}{2}\right)^2P_{\ell-1}(x) = \left(x-\frac{m}{2}\right)P_\ell(x)-\frac{\ell(m+1-\ell)}{4}P_{\ell-1}(x).
\end{equation}

It suffices to show that $k_\ell$ obeys the same three-term recursion. This follows from the generating function of $K_\ell(x)$. Recall that $G(x, z)=\sum_{\ell=0}^m K_\ell(x)z^\ell$  from \cref{eq:kravchuk-gen-func}.  Then
\begin{equation}
\frac{\partial G}{\partial z} = \left(\frac{m-x}{1+z}-\frac{x}{1-z}\right)G(x, z)\implies (1-z^2)\frac{\partial G}{\partial z}  = (m-mz-2x)G(x, z).
\end{equation}
We write $G(x, z)=\sum_{\ell=0}^m K_\ell(x)z^\ell$ and extract the $[z^\ell]$-coefficient to obtain three-term recursion
\begin{equation}
(\ell+1)K_{\ell+1}(x)=(m-2x)K_\ell(x)-(m-\ell+1)K_{\ell-1}(x).
\end{equation}
Finally, applying the scaling between $K_\ell$ and $k_\ell$ in \cref{eq:kravchuk-def}, we see that $k_\ell$ obeys \cref{eq:three-term-recursion}.
\end{proof}

The formulation in \cite{dqi} makes it clear that the limitation in the semicircle law comes from the fact that we cannot take $\ell$ larger than $d^{\perp}/2$ and still decode $C^\perp$ without error.  For $\ell > d^{\perp}/2$, we introduce error, analogous to the loss of $k$-wise independence for $k > n$ in our moments-based proof of \cref{prop:dqi-on-opi}.
 We will return to this point in \cref{sec:motivation}, when we improve upon the semicircle law by going beyond this barrier.
\subsection{The Discrepancy Proof}
\label{sec:discrepancy}
Having already explained two proofs of \cref{prop:dqi-on-opi} (the one in \cite{dqi} and the moments-based proof above), we now give a third.  This third proof is a streamlined version of the proof in \cite{dqi} on the performance guarantee of DQI. The reason that we give this third proof is because it will help us set up the analytical framework we need to improve beyond the semicircle law in \cref{sec:improvements}.  In particular, in this proof we black-box out the Fourier analysis in \cite{dqi} into \cref{lem:key-fourier}, which simplifies the derivation of \cref{eq:dqi-quadratic-form} and allows us to bypass the step of controlling eigenvalues of the matrix $A^{(\ell)}$.

Motivated by the state DQI samples from in \cref{eq:dqi-state}, 
we begin by making the following definition. 
\begin{definition}
Recall $f_i$ from \cref{eq:fi}.
Define \emph{$k$-wise discrepancy} $q_k:\mb{F}_q^n\to \mb{Z}$ by
\begin{equation}
    q_k(x):= \sum_{S\in \binom{[m]}{k}} \prod_{i\in S}f_i(\langle b_i, x\rangle).
\end{equation}
\end{definition}

All the Fourier analysis in \cite{dqi} can be black-boxed into one step, which we present in \cref{lem:balanced-fourier} below. We omit the proof as it is a special case of \cref{lem:key-fourier} (proved later) and also implicit in \cite{dqi}.
\begin{lemma}
\label{lem:balanced-fourier}
$\mb{E}[q_0(x)]=1$ and $\mb{E}[q_k(x)] = 0$ for all $1\le k< d^\perp$, where the expectation is over $x\sim\on{Unif}(\mb{F}_q^n)$.
\end{lemma}

Note that, the satisfaction ratio of a solution to the Max-LINSAT problem (\cref{question:max-linsat}) is
\begin{equation}\label{eq:bal-s-into-q1}
s(x)=\sum_{i=1}^m \frac{1+f_i(\langle b_i, x\rangle)}{2m} = \frac{1}{2}+\frac{1}{2m}q_1(x),
\end{equation}
so \cref{prop:dqi-on-opi} reduces to finding a distribution $\mb{P}$ such that as $m, n\to \infty$ with $n\sim 2\mu m$
\begin{equation}
    \mb{E}_{x\sim \mb{P}}\left[q_1(x)\right]\ge 2m\sqrt{\mu(1-\mu)}-o(m).
\end{equation}

For weights $u\in \mb{R}^{\ell+1}$ optimized later, DQI samples solution $x$ to Max-LINSAT with probability
\begin{equation} 
\label{eq:dqi-Pu}
\mb{P}_u(x) \propto \left(\sum_{k=0}^\ell u_k q_k(x)\right)^2,
\end{equation}
where $2\ell+1 < d^\perp = n+1$. Relative to \cref{eq:dqi-state}, we have changed the weights $w$ to weights $u$ that are normalized differently\footnote{We note that our weights $u_k$ are not the same as the weights $u_k$ in \cite{dqi}, due to the difference in Fourier transform conventions.}:  We have
\begin{equation}
\label{eq:u-w-normalization}
w_k = \binom{m}{k}^{1/2}u_k.
\end{equation}
 We henceforth forget about the weights $w$ and use the normalized versions $u$ instead.
Next, we explain how to handle products of $q_k$ necessary when we expand \cref{eq:dqi-Pu}. We start with some notation.
\begin{definition}\label{def:Nkt}
Extending the notion of {symmetric difference} to multiple arguments, we define
$\Delta(T_1, \ldots, T_r)$ to be the set of elements in an odd number of $T_i$'s. For $k_i, t\in [m]$, define 
\begin{equation}\label{eq:def-Nkt}
    N(k_1, \ldots, k_r; t) \coloneqq \#\left\{(T_1, \ldots, T_r): T_i\in \binom{[m]}{k_i}, \Delta (T_1, \ldots, T_r) = [t]\right\}
\end{equation}
\end{definition}
\begin{proposition}
\label{prop:bal-prod}
For any $k_1, \ldots, k_r\in \mb{N}$, we have that
\begin{equation}
\mb{E}_{x\in \mb{F}_p^n} \left[q_{k_1}(x) \cdot \cdots \cdot q_{k_r}(x)\right] = \sum_{t=0}^{m} N(k_1, \ldots, k_r;t) \mb{E}_{x\in \mb{F}_p^n} [q_t(x)].
\end{equation}
\end{proposition}
\begin{proof}
Let $\bigsqcup$ denote the disjoint multi-set union.
As $f_i$ are $\pm 1$-valued, we have
\begin{equation}
\begin{aligned}
\mb{E}\left[q_{k_1}(x)\ldots q_{k_r}(x)\right] & =\mb{E}\left[\sum_{S_i\in\binom{[m]}{k_i}} \prod_{i\in S_1\bigsqcup \ldots \bigsqcup S_r}f_i(\langle b_i, x\rangle)\right]
\\ & =\mb{E}\left[\sum_{T\subset[m]}N(k_1, \ldots, k_r;|T|)\prod_{i\in T}f_i(\langle b_i, x\rangle)\right]
\\ & = \sum_{t=0}^{m}\mb{E}\left[\sum_{T\in\binom{[m]}{t}} N(k_1, \ldots, k_r;t)\prod_{i\in T}f_i(\langle b_i, x\rangle) \right]
\\ & = \sum_{t=0}^{m} N(k_1, \ldots, k_r;t) \mb{E} [q_t(x)]
\end{aligned} 
\end{equation}
as desired.
\end{proof}
Now, we have all the ingredients to give our streamlined account of the performance of DQI. As $k+k'+1\le 2\ell+1 < d^\perp$, and recalling the definition of $\mathbb{P}_u$ from \cref{eq:dqi-Pu}, we compute by \cref{prop:bal-prod} and \cref{lem:balanced-fourier} that
\begin{equation}
\label{eq:bal-quad-form-q1}
\mb{E}_{x\sim \mb{P}_u} [q_1(x)] = \frac{\sum_{k, k'=0}^\ell u_ku_{k'}\mb{E}[q_k(x)q_{k'}(x)q_1(x)]}{\sum_{k, k'=0}^\ell u_ku_{k'}\mb{E}[q_k(x)q_{k'}(x)]}
=\frac{\sum_{k, k'=0}^\ell u_ku_{k'}N(k,k',1; 0)}{\sum_{k, k'=0}^\ell u_ku_{k'}N(k,k'; 0)}.
\end{equation}

At this point, one way to conclude \cref{prop:dqi-on-opi} is as follows: Notice that 
\cref{eq:bal-quad-form-q1} is equivalent to \cref{eq:dqi-quadratic-form} upon expressing both quadratic forms with $w$. We see this calculation later in \cref{prop:step-0}. Then, we can apply \cref{prop:dqi-on-opi} to maximize the quadratic forms and recover $\on{SCL}_{1/2}(\mu)$.

Here, we take a different approach and bypass the analysis of matrix $A^{(\ell)}$ in both the eigenvalue control as in \cite{dqi} and the analogous characteristic polynomial control in \cref{prop:3-term-rec}.
\begin{lemma}\label{lem:Nzero}
If $k_1+\ldots +k_r$ is odd, then $ N(k_1, \ldots, k_r; 0)=0$. Otherwise,
\begin{equation}
 N(k_1, \ldots, k_r; 0) = 2^{-m}\sum_{t=0}^m \binom{m}{t} \prod_{i=1}^r K_{k_i}(t).
\end{equation}
\end{lemma}
\begin{proof}
The odd sum case is trivial by parity. When the sum of the $k_i$ is even, we compute the generating function for $N(k_1, \ldots, k_r;0)$. Define
$F\in \mb{Z}[X_1, \ldots, X_r]$ by
\begin{equation}\label{eq:F}
F(X_1, \ldots, X_r) \coloneqq \sum_{\varepsilon\in\mc{E}}\prod_{i=1}^r\prod_{j=1}^mX_i^{\varepsilon_{ij}}\quad\text{where}\quad \mathcal{E} \coloneqq \left\{ \varepsilon \in \{0,1\}^{r\times m} \,:\, 2 \, \Big\vert \,\sum_{i=1}^r\varepsilon_{ij} \ \forall j \in [m] \right\}.
\end{equation}
For any $\varepsilon \in \{0,1\}^{r\times m}$, let $T_i = \{ j \in [m] \,:\, \varepsilon_{ij} = 1\}$. Observe that $\varepsilon \in \mathcal{E}$, i.e. $\sum_{i=1}^r \varepsilon_{ij}$ is even for all $j$, if and only if $\Delta(T_1, \ldots, T_r) = \emptyset$.
Thus, upon expanding the right hand side of \cref{eq:F}, the coefficient of $X_1^{k_1}\ldots X_r^{k_r}$ is exactly the number of sets $T_1, \ldots, T_r$ such that $|T_i|=k_i$ and $\Delta(T_1, \ldots, T_r)=\emptyset$, which is exactly 
$N(k_1, \ldots, k_r; 0)$. Therefore, $F$ is the generating function of $N(k_1, \ldots, k_r; 0)$, i.e.
\begin{equation}\label{eq:Nkk-gen-func}
    F(X_1, \ldots, X_r) = \sum_{k_1, \ldots, k_r} N(k_1, \ldots, k_r; 0)X_1^{k_1}\ldots X_r^{k_r}.
\end{equation}
Then, we factor over $j$ and observe the sum is independent of $j$, so 
\begin{equation}
F(X_1, \ldots, X_r) = \prod_{j\in [m]}\left(\sum_{\substack{(\varepsilon_{1j}, \ldots, \varepsilon_{rj})\\2\mid \sum_{i=1}^r\varepsilon_{ij}}}\prod_{i=1}^r X_i^{\varepsilon_{ij}}\right)
=\left[\frac{1}{2}\prod_{i=1}^r (1+X_i)+\frac{1}{2}\prod_{i=1}^r (1-X_i)\right]^m.
\end{equation}
Let $[\prod_i X_i^{k_i}]F$ denote the coefficient of $\prod_i X_i^{k_i}$ in $F$. By \cref{eq:Nkk-gen-func}, we obtain
\begin{equation}
\begin{aligned}
N(k_1, \ldots, k_r; 0) & =\left[\prod_i X_i^{k_i}\right]F(X_1, \ldots, X_r)
\\ & = 2^{-m}\sum_{t=0}^m \binom{m}{t} \prod_{i=1}^r \left[X_i^{k_i}\right](1-X_i)^t(1+X_i)^{m-t}
\\ & = 2^{-m}\sum_{t=0}^m \binom{m}{t} \prod_{i=1}^r K_{k_{i}}(t)
\end{aligned}
\end{equation}
using the generating function of $K_{\ell}(t)$ given in \cref{eq:kravchuk-gen-func}.
\end{proof}

Applying \cref{lem:Nzero} to \cref{eq:bal-quad-form-q1}, and using that $K_1(t)=m-2t$, we have that 
\begin{equation}\label{eq:pre-KKT}
\mb{E}_{x\sim \mb{P}_u} [q_1(x)] =
\frac{\sum_{t=0}^m \binom{m}{t}(m-2t)X_u(t)^2}{\sum_{t=0}^m \binom{m}{t}X_u(t)^2} \quad\text{where}\quad X_u(t):=\sum_{k=0}^\ell u_k K_k(t).
\end{equation}
We view this as the expectation of $m-2T$ for random variable $T$ with $\mb{P}(T=t)\propto \binom{m}{t}X_u(t)^2$. To maximize satisfaction fraction, we choose $u$ to minimize $\mb{E}T$. We comput this minimum.

\begin{lemma}
\label{lem:KKT}
Over all $u\in\mb{R}^{\ell+1}$, minimum value of $\mb{E}T$ with $\mb{P}(T=t)\propto \binom{m}{t}X_u(t)^2$ is $Z_{\min}(K_{\ell+1})$.
\end{lemma}
\begin{proof}
Let $u_\star$ be the minimizer of $\mb{E}T$. By stationarity, we have that
\begin{equation}
0 = \frac{\partial[\mb{E}T]}{\partial u_j}\bigg\vert_{u=u_\star}=\sum_{t=0}^m (t-\gamma)\binom{m}{t}X_{u_\star}(t)K_j(t)
\end{equation}
for each $0\le j\le \ell$. Hence $(t-\gamma)X_{u_\star}(t)$ is orthogonal to $K_0, \ldots, K_\ell$ under the measure induced by $\on{Bin}(m, 1/2)$ and it has degree at most $\ell+1$ as $\deg X_u\le \ell$. This implies that $(t-\gamma)X_u(t)$  must be a multiple of the next Kravchuk polynomial $K_{\ell+1}$. Thus, $\mb{E}T\ge Z_{\min}(K_{\ell+1})$ for every $u$.

Conversely, as $K_0, \dots, K_\ell$ is an orthonormal basis of polynomials of degree at most $\ell$ under the inner product in \cref{eq:orth-rel} for the $\sigma=\on{Bin}(m, 1/2)$ distribution, there exists $u\in\mb{R}^{\ell+1}$ such that
\[
\frac{K_{\ell+1}(t)}{t-Z_{\min}(K_{\ell+1})} = \sum_{k=0}^\ell u_k K_k(t)=X_u(t).
\]
since the left hand side has degree at most $\ell$. Recalling \cref{eq:orth-rel}, we have that
\[
\mb{E}T -Z_{\min}(K_{\ell+1}) =  \frac{\sum_{t=0}^m (t-Z_{\min}(K_{\ell+1}))\binom{m}{t}X_u(t)^2}{\sum_{t=0}^m \binom{m}{t}X_u(t)^2}
=\frac{\langle K_{\ell+1}, X_u\rangle_\sigma}{\langle X_u, X_u\rangle_\sigma}
= \sum_{k=0}^\ell u_k\frac{\langle K_{\ell+1}, K_k\rangle_\sigma}{\langle X_u, X_u\rangle_\sigma}
=0.
\]
Therefore, there exists $u$ such that $\mb{E}T$ is exactly the smallest root of $K_{\ell+1}$.
\end{proof}

Now, we compute this smallest root and \cref{prop:dqi-on-opi} follows.
\begin{proof}[Third Proof of \cref{prop:dqi-on-opi}]
Combining \cref{eq:bal-s-into-q1,eq:pre-KKT}, we see that for any $u\in\mb{R}^{\ell+1}$, and with $T$ distributed according to \cref{lem:KKT} parametrized by $u$, that
\begin{equation}
\label{eq:mET}
\mb{E}_{x\sim \mb{P}_u}[s(x)] = 1-\frac{1}{m}\mb{E} T\le 1-\frac{Z_{\min}(K_{\ell+1})}{m}.
\end{equation}
Now, $K_{\ell+1}(x)=(-1)^{\ell+1}K_{\ell+1}(m-x)$ by \cref{eq:kravchuk-def}, so $Z_{\min}(K_{\ell+1})=m-Z_{\max}(K_{\ell+1})$. Applying \cref{fact:Nki-kravchuk} to \cref{eq:mET}, we see that
\begin{equation}
\label{eq:mETT}
\mb{E}_{x\sim \mb{P}_u}[s(x)] \le \frac{Z_{\max}(K_{\ell+1})}{m} \to \on{SCL}_{1/2}(\mu),
\end{equation}
recovering \cref{prop:dqi-on-opi}.
\end{proof}
We remark that the proof of \cref{prop:dqi-on-opi} uses only one of the directions of \cref{lem:KKT}. The converse direction suggests a barrier at the semi-circle law: Namely, in \cref{eq:bal-s-into-q1}, any polynomial $X_u$ of degree at most $\ell$ can do no better. The bottleneck is essentially the requirement that $\ell < d^\perp/2$ due to \cref{lem:balanced-fourier}. If we could increase $\ell$ while preserving \cref{eq:bal-quad-form-q1}, then choosing $X_u$ to be a higher degree Kravchuk polynomial would yield improvements beyond the semi-circle law.

\subsection{The Semicircle Law Barrier}
\label{sec:motivation}
We discuss why all the proofs encounter the semicircle law barrier, which motivates our approach in \cref{sec:improvements} to improve beyond it.
Observe that \cref{lem:balanced-fourier} is the key fact in the DQI proof in \cite{dqi} and the discrepancy proof; while the moments proof relies on the fact that random variable $X=m\cdot s(x)$, for a uniform random solution $x$, has the same first $n$ moments as $\on{Bin}(m, 1/2)$.

These two facts are equivalent. Although $x\sim \mb{P}_u$ is not uniform in \cref{lem:balanced-fourier}, the symmetric form \cref{eq:dqi-Pu} means $\mb{E}[q_k(x)]$ is a linear combination of the first $k$ moments of $X$. Then, both facts are $n$ equations on the first $n$ moments of $X$ that is satisfied by $\on{Bin}(m, 1/2)$, and must be equivalent.

To beat the semicircle law, we need to sample $x$ not of the form in \cref{eq:dqi-Pu}. One naive idea is to sample $\mb{P}(x)\propto \left(\sum_{k=0}^\ell u_kq_k(x)\right)^4$. Controlling terms in \cref{prop:bal-prod} with \cref{lem:balanced-fourier} requires setting $4\ell < d^\perp$. However, it turns out we observe empirically that even the optimal $u$ cannot improve beyond the semicircle law of DQI. In light of the discussion above, falling short of the semicircle law without going beyond $d^\perp$ is not surprising. Consider the following question.

\begin{question}
Are there sets $S_i$ such that $s(x)$ for $x\sim\on{Unif}(\mb{F}_p^n)$ converges in distribution as $m\to\infty$ to $Z/m$ where $Z$ is the principal representation in \cref{thm:cms} associated to $\on{Bin}(m, 1/2)$?
\end{question}

Intuitively, the answer to this question seems to be no: it seems obvious that there is no way to choose $S_i$ such that the number of satisfied constraints takes at most $\sim n/2$ values as we vary $x$, but we are not aware of a proof.
Without a negative answer to the question, it could be that the maximum satisfaction is actually the semicircle law. Therefore, any successful improvement needs to distinguish $s(x)$ and $Z/m$ further, beyond looking at the first $n$ moments.

\section{Beyond the Semicircle Law}
\label{sec:improvements}
In this section, we set up how we improve beyond the semicircle law. Our approach is a combination of techniques from works on local leakage resilience of Shamir's secret sharing scheme, combined with our framework in \cref{sec:discrepancy} extended to handle beyond the minimum distance. 
\subsection{The Generalized Framework}
\label{sec:framework}
We begin by generalizing much of the discrepancy proof in \cref{sec:discrepancy} to the case when input lists $S_i\subset\mb{F}_p$ have size $\rho p$ for $\rho \in (0, 1)$ which in particular resolves the balanced case $|S_i|\sim p/2$. We will choose $\ell \coloneqq (\mu + \delta)m$ for some $\delta >0$ and define $\mb{P}_u$ as in \cref{eq:dqi-Pu} with this choice of $\ell$. We aim to show
\begin{equation}
\mb{E}_{x \sim \mb{P}_u} \left[s(x)\right]\ge \on{SCL}_\rho(\mu+\delta)-o(1)\quad\text{where}\quad \on{SCL}_\rho(\alpha) \coloneqq \begin{cases}
\left(\sqrt{\alpha(1-\rho)}+\sqrt{\rho(1-\alpha)}\right)^2 &\text{if }\alpha+\rho \le 1\\
1 & \text{if }\alpha+\rho \ge 1.
\end{cases}
\end{equation}

For the general $\rho$ case, we instead work with a normalized version $g_i$ of $f_i$, such that $\mb{E}[g_i] = 0$ and $\on{Var}[g_i]=1$. We define $k$-wise discrepancy $q_k$ from it, namely
\begin{equation}\label{eq:def-g}
q_k(x)\coloneqq \sum_{S\in\binom{[m]}{k}} \prod_{i\in S} g_i(\langle b_i, x\rangle)\quad\text{where}\quad g_i(x_i)=\frac{f_i(x_i)-\mb{E}[f_i]}{\sqrt{\on{Var}[f_i]}}=\frac{f_i(x_i)-(2\rho-1)}{2\sqrt{\rho(1-\rho)}}.
\end{equation}
First, as before, with expectation over distribution $x\sim \mb{P}_u$, we can write the objective as 
\begin{equation}\label{eq:s-to-q1}
    \mb{E}s(x) = \frac{1}{m}\sum_{i=1}^m  \mb{E}\left[\frac{f_i(\langle b_i, x\rangle)+1}{2}\right]
= \frac{1}{m}\sum_{i=1}^m  \mb{E}\left[\rho+\sqrt{\rho(1-\rho)}g_i(\langle b_i, x\rangle)\right]
= \rho + \frac{\sqrt{\rho(1-\rho)}}{m} \mb{E}[q_1(x)].
\end{equation}
Note that in the perfectly balanced case where $\rho =1/2$, we have $g_i = f_i$, we recover the same definition of $q_k$, and this agrees with \cref{eq:bal-s-into-q1}. For the general $\rho$ case, defining $q_k$ based on $g_i$ enables the following generalization of \cref{lem:balanced-fourier}.
\begin{lemma}
\label{lem:key-fourier}
Let $C_t^\perp\subset C^\perp$ be codewords $y\in C^\perp$ with Hamming weight $t$. Over uniform $x\in\mb{F}_p^n$
\begin{equation}
\mb{E}[q_t(x)] = \sum_{y\in C_t^\perp}\prod_{i:y_i\ne 0} \wh{g}_i(y_i).
\end{equation}
\end{lemma}
\begin{proof}
Temporarily let $\wh{g}(y)=\prod_{i:y_i\ne 0} \wh{g}_i(y_i)$. 
By definition of $q_t$, we compute that
\begin{equation}
\begin{aligned}
\mb{E} \left[q_t(x) \right] & = \sum_{T\in\binom{[m]}{t}} \mb{E}\left[\prod_{i\in T}g_i\left(\langle b_i, x\rangle\right)\right]
\\ & = \sum_{T\in\binom{[m]}{t}} \mb{E} \left[\prod_{i\in T}\left(\sum_{y_i\in \mb{F}_p} \wh{g_i}(y_i) e_p\left(-y_i\langle b_i, x\rangle\right)\right)\right]
\\ & = \sum_{T\in\binom{[m]}{t}} \mb{E} \left[\sum_{y\in \mb{F}_p^m:\on{supp} y =T} \wh{g}(y)e_p\left(-\sum_{i\in T}y_i\langle b_i, x\rangle\right)\right]
\\ & =\sum_{y\in \mb{F}_p^m:|y|=t}\wh{g}(y) \mb{E}\left[e_p\left(-\langle B^\top y, x\rangle\right)\right],
\end{aligned}
\end{equation}
where in the last step, note that if $b_i$ is the $i$-th row of $B$ then
\begin{equation}
\sum_{i:y_i\ne 0}y_i\langle b_i, x\rangle
=\sum_{i=1}^m\sum_{j=1}^n y_i B_{ij}x_j = \langle B^\top y, x\rangle.
\end{equation}
Now, note that $\mb{E}[e_p(\langle z, x\rangle)]=\mbf{1}\{z=0\}$ for any $z \in \mb{F}_p^n$, so we get
\begin{equation}
\mb{E} \left[q_t(x) \right]=\sum_{y\in \mb{F}_p^m:|y|=t}\wh{g}(y) \mbf{1}\{B^Ty=0\}
=\sum_{y\in C_t^\perp}\wh{g}(y),
\end{equation}
by definition of $C^\perp_t$.
\end{proof}
The proof of \cref{prop:bal-prod}, which expresses $\mathbb{E}[ q_{k_1}(x)\cdots q_{k_r}(x)]$ in terms of $N(k_1, \ldots, k_r; t)$ does not hold if we replace the $f_i$ with the $g_i$, as it used the fact that the $f_i$ were $\pm 1$-valued, and hence $f_i(a)^2 = 1$ for all $a$. To recover something similar in spirit, 
the key analogous observation is the following interpolation/degree reduction step using the fact that $g_i$ takes \emph{two} values.
\begin{lemma}
For any $i\in [m]$,
\begin{equation}\label{eq:beta}
g_i(a)^2 = 1+\beta g_i(a)\quad\text{where}\quad\beta \coloneqq \frac{1-2\rho}{\sqrt{\rho(1-\rho)}}.
\end{equation}
\end{lemma}
\begin{proof}
If $a\in S_i$, then $g_i(a) = \sqrt{(1-\rho)/\rho}$, and 
\[
1+\beta g_i(a) = 1+ \frac{1-2\rho}{\sqrt{\rho(1-\rho)}}\cdot \sqrt{\frac{1-\rho}{\rho}}=\frac{1-\rho}{\rho}
=g_i(a)^2.\]
If $a\not\in S_i$, then $g_i(a) = -\sqrt{\rho/(1-\rho)}$, and 
\[
1+\beta g_i(a) = 1- \frac{1-2\rho}{\sqrt{\rho(1-\rho)}}\cdot \sqrt{\frac{\rho}{1-\rho}}=\frac{\rho}{1-\rho}=g_i(a)^2,
\]
as desired.
\end{proof}
This allows us to recover an analog of \cref{prop:bal-prod}.
\begin{proposition}\label{prop:gen-prod}
In expectation over $x\sim\on{Unif}(\mb{F}_p^n)$, we have that
\begin{equation}\label{eq:def-N-rho}
\mb{E}\left[q_k(x)q_{k'}(x)\right] = \sum_{t=0}^{m} N_\rho(k, k';t)\mb{E}[q_t(x)]\quad\text{where}\quad N_\rho(k, k';t) = \sum_{j=0}^t\beta^{j}\binom{t}{j}\binom{t-j}{\frac{t+k-k'-j}{2}}\binom{m-t}{\frac{k+k'-t-j}{2}},
\end{equation}
where we adopt the convention that binomial coefficients are zero if the arguments are not integers.
\end{proposition}
\begin{proof}
We do a similar set-counting argument as in the proof of \cref{prop:bal-prod}.
\begin{equation}\label{eq:counting}
\begin{aligned}
\mb{E}\left[q_k(x)q_{k'}(x)\right] & = \sum_{|T|=k, |T'|=k'} \mb{E}\left[\prod_{i\in T\sqcup T'} g_i(\langle b_i, x\rangle)\right]
\\ & = \sum_{|T|=k, |T'|=k'} \mb{E}\left[\prod_{i\in T\Delta T'} g_i(\langle b_i, x\rangle)\cdot \prod_{i\in T\cap T'} g_i(\langle b_i, x\rangle)^2 \right]
\\ & = \sum_{|T|=k, |T'|=k'} \mb{E}\left[\prod_{i\in T\Delta T'} g_i(\langle b_i, x\rangle)\cdot \prod_{i\in T\cap T'} \left(1+\beta g_i(\langle b_i, x\rangle)\right) \right]
\\ & = \sum_{|T|=k, |T'|=k'} \mb{E}\left[\sum_{J\subset T\cap T'}\beta^{|J|}\prod_{i\in (T\Delta T')\cup J } g_i(\langle b_i, x\rangle) \right]
\\ & = \sum_{U\subset [m]}
\left(\sum_{|T|=k, |T'|=k'} 
\mbf{1}\left\{T\Delta T'\subset U\subset T\cup T'\right\} \beta^{|U|-|T\Delta T'|} \right)\mb{E}\left[\prod_{i\in U}g_i(\langle b_i, x\rangle) \right]
\\ & = \sum_{t=0}^m \mb{E}[q_t(x)] \left(\sum_{j=0}^t \beta^{j} \cdot \#\left\{(T, T')\in\binom{[m]}{k}\times\binom{[m]}{k'}:T\Delta T'\subset [t]\subset T\cup T', |T\Delta T'|=t-j\right\} \right)
\end{aligned}
\end{equation}
where we note that the term in the parenthesis on the penultimate line depends on $U$ only through $|U|=t$, so we let $U=[t]$, and $|T\Delta T'|=t-j$ for some $0\le j\le t$. 
To show \cref{eq:def-N-rho}, we need to show the number of pairs $(T, T')$ that are counted in the final line of \cref{eq:counting} is 
\begin{equation}\label{eq:three-binom}
 \binom{t}{j}\binom{t-j}{\frac{t+k-k'-j}{2}}\binom{m-t}{\frac{k+k'-t-j}{2}}.
\end{equation}
There are $\binom{t}{j}$ ways to choose $T\Delta T'$ of size $t-j$ from $[t]$. We know the other $j$ elements of $[t]$ must lie in $T\cap T'$, which has size $(|T|+|T'|-|T\Delta T'|)/2 = (k+k'-t+j)/2$. We choose the remaining $(k+k'-t-j)/2$ elements of $T\cap T'$ from $[m]\setminus [t]$, which gives the last binomial coefficient in \cref{eq:three-binom}. Finally, for each fixed $T\Delta T'$ of size $t-j$, we split it into $T\setminus T'$ and $T'\setminus T$ of sizes $(t+k-k'-j)/2$ and $(t-k+k'-j)/2$, respectively, giving the middle binomial coefficient in  \cref{eq:three-binom}. 

Note that if any of the set sizes computed above are not integral, then no such $(T, T')$ exists, agreeing with our binomial coefficient convention.
\end{proof}
Note that when $\rho=1/2$, then $\beta = 0$ and only the $j=0$ term contributes to $N_{1/2}(k,k';t)$ and it is exactly $N(k,k';t)$ from \cref{def:Nkt}. Observe that by definition, $N_\rho(k, k';s)=0$ for $s> k+k'$, and $\mb{E}[q_s(x)]=0$ for $0<s\le k+k'<d^\perp$. Therefore, the computation in \cref{prop:gen-prod} simplifies significantly below the $d^\perp$ cutoff.
\begin{corollary}
If $k+k'< d^\perp$ for $k, k'\ge 0$, then
\begin{equation}
\mb{E}[q_kq_k'] = N_\rho(k, k';0)\mb{E}[q_0]=\binom{m}{k}\mbf{1}\{k=k'\}.
\end{equation}
\end{corollary}
Moreover, we derive three-term recursion analogous to \cref{prop:3-term-rec}, as follows.
\begin{proposition}\label{prop:q-triple-rec}
Recall the definition of $\beta$ from \cref{eq:beta}. For any $0\le k<m$,
\begin{equation}\label{eq:triple-rec}
q_1q_k  = (k+1)q_{k+1}+(m-k+1)q_{k-1}+\beta kq_k.
\end{equation}
\end{proposition}
\begin{proof}
Recall the definition of $q_1$ and $q_k$ from \cref{eq:def-g}.
In the definition of $q_1$, there is a summation over a singleton set; call that $\{j\}$. Let the $k$-element set summed over in the definition of $q_k$ be $T$. We split the sum into cases depending on if $j\in T$ and apply the fact that $g_j^2 = 1 + \beta g_j$ to obtain

\begin{equation}\label{eq:three-summands}
\begin{aligned}
q_1 q_k & = \sum_{j\in [m]}\sum_{T\in\binom{[m]}{k}}
\prod_{i\in T\sqcup\{j\}}g_i 
\\ & = \sum_{T\in\binom{[m]}{k}}\left(\sum_{j\not\in T} g_j
\prod_{i\in T}g_i + \sum_{j\in T} g_j^2
\prod_{i\in T\setminus\{j\}}g_i \right)
\\ & = \sum_{T\in\binom{[m]}{k}}\left( \sum_{j\not\in T} g_j
\prod_{i\in T}g_i+\sum_{j\in T}
\prod_{i\in T\setminus\{j\}}g_i + \beta\sum_{j\in T}\prod_{i\in T}g_i \right).
\end{aligned}
\end{equation}
We examine the three summands in order:
\begin{itemize}
    \item If $j\not\in T$, then $T\cup \{j\}$ is a $k+1$ element subset of $[m]$, and each such set is counted $k+1$ times in the first summand of \cref{eq:three-summands} for choosing~$j$. This gives the first term of \cref{eq:triple-rec}.
    \item If $j\in T$, then $T\setminus \{j\}$ is a $k-1$ element subset of $[m]$, and each such set is counted $m-k+1$ times in the second summand of \cref{eq:three-summands} for choosing~$j$. This gives the second term of \cref{eq:triple-rec}. 
    \item For the last term, the sum over $j\in T$ gives a factor of $k$, so it matches the last term of \cref{eq:triple-rec}.
\end{itemize}
Together, we obtain \cref{eq:triple-rec}.
\end{proof}
Therefore, combining the previous two propositions, we obtain the following corollary.
\begin{corollary}
\label{cor:triple-q}
Let $
N_\rho(k, k', 1;s)\coloneqq (k+1)N_\rho(k+1, k';s)+\beta k N_\rho(k, k';s)+(m-k+1)N_\rho(k-1, k';s)$. Then
\begin{equation}
\mb{E}_{x\in\mb{F}_p^n}\left[q_1(x)q_k(x)q_{k'}(x)\right] = \sum_{s=0}^{m} N_\rho(k, k',1;s)\mb{E}_{x\in\mb{F}_p^n}[q_s(x)].
\end{equation}
\end{corollary}

\subsection{The Expected Satisfaction Expansion}
\label{sec:expansion}
The idea from our framework in \cref{sec:discrepancy} is to sample $x$ with probability given as the the square of a suitably weighted linear combination of the $k$-wise discrepancy for $0\le k\le \ell$, similar to the quantum state of DQI. As discussed, owing to the easy control of expected discrepancy below $d^\perp$ in \cref{lem:balanced-fourier}, we ensured $\ell\le \mu m$, and this is the bottle-neck that blocks us from improving beyond the semicircle law. Instead, we take 
\begin{equation}\label{eq:def-delta}
    \ell \coloneqq \lfloor (\mu+\delta)m \rfloor \quad\text{where}\quad 0<\mu\le \mu+\delta \le 1-\rho.
\end{equation}
for some $\delta\ge 0$ to be chosen. The cutoff of $1-\rho$ for $\mu+\delta$ is chosen because we hope to find a solution $x$ with $s(x)\ge \on{SCL}_\rho(\mu+\delta)-o(1)$, which is $1-o(1)$ once $\mu+\delta \ge 1-\rho$, so there is no need to take larger $\delta$. 
As in the perfect decoding case, we choose $u$ and $x\sim \mb{P}_u$ exactly as before:
\begin{equation}
\mb{P}_u(x) \propto \left(\sum_{k=0}^\ell u_k q_k(x)\right)^2\quad\text{where}\quad u_k \propto \binom{m}{k}^{-1/2}\mbf{1}\{\ell - \sigma \le k \le \ell\},
\end{equation}
and $\sigma \gg 1$ grows slowly with $\ell$ and so with $m$ and $n$.
For now, we fix
\begin{equation}
\label{eq:def-sigma}
\sigma \coloneqq \lfloor \log\log\ell\rfloor \sim \log\log m.
\end{equation}
In \cite[Lemma 9.3]{dqi}, the weight vector $w$ is chosen to be the leading eigenvector of $A^{(\ell)}$, which turns out to be exactly in this form for $\sigma=\sqrt{\ell}$ and the correct normalization \cref{eq:u-w-normalization}. 
Analogous to \cref{eq:bal-quad-form-q1}, we use shorthand $\mb{E}[q_t]=\mb{E}_{x\in\mb{F}_p^n}[q_t(x)]$ to compute that
\begin{equation}
\label{eq:master-t}
\begin{aligned}
\mb{E}_{\mb{P}_u}[q_1(x)] & = \frac{\sum_{k, k'=0}^\ell u_ku_{k'} \mb{E}[q_1q_kq_{k'}]}{\sum_{k, k'=0}^\ell u_ku_{k'} \mb{E}[q_kq_{k'}]}
\\ & = \frac{\sum_{t=0}^{m} \mb{E}[q_t]  \sum_{k, k'=0}^\ell u_ku_{k'}  N_\rho(k, k', 1;t)}{\sum_{t=0}^{m} \mb{E}[q_t]  \sum_{k, k'=0}^\ell u_ku_{k'} N_\rho(k, k';t)}
\\ & = \frac{\sum_{k, k'=\ell-\sigma}^\ell \binom{m}{k}^{-1/2}\binom{m}{k'}^{-1/2}  N_\rho(k, k', 1;0) + \sum_{t=d^\perp}^{m} \mb{E}[q_t]  \sum_{k, k'=\ell-\sigma}^\ell \binom{m}{k}^{-1/2}\binom{m}{k'}^{-1/2}  N_\rho(k, k', 1;t)}{\sum_{k, k'=\ell-\sigma}^\ell \binom{m}{k}^{-1/2}\binom{m}{k'}^{-1/2}  N_\rho(k, k';0)+\sum_{t=d^\perp}^{m} \mb{E}[q_t]  \sum_{k, k'=\ell-\sigma}^\ell \binom{m}{k}^{-1/2}\binom{m}{k'}^{-1/2} N_\rho(k, k';t)}.
\end{aligned}
\end{equation}
We separate the terms by the weight $t$. 
The $t=0$ terms in the numerator and denominator of \cref{eq:master-t} can be computed similarly to \cref{eq:bal-quad-form-q1}, and their ratio together with \cref{eq:s-to-q1} will essentially recover $\on{SCL}_\rho(\mu+\delta)$, which strictly improves on $\on{SCL}_\rho(\mu)$.
The summands with $0<t<d^\perp$ terms vanish via \cref{lem:key-fourier}, exactly as before. We also have some higher order correction terms for $t\ge d^\perp$, which we will have to control.

In the following proposition we use the last two corollaries to compute the $t=0$ terms.
\begin{proposition}\label{prop:step-0}
It holds that
\begin{equation}
\begin{aligned}
\sum_{k, k'=\ell-\sigma}^\ell \binom{m}{k}^{-1/2}\binom{m}{k'}^{-1/2}  N_\rho(k, k';0)&   = \sigma+1,
\\ \sum_{k, k'=\ell-\sigma}^\ell \binom{m}{k}^{-1/2}\binom{m}{k'}^{-1/2}  N_\rho(k, k', 1;0) & = \left[\beta  (\mu+\delta) + 2\sqrt{(\mu+\delta)(1-\mu-\delta)}\right]\left(\sigma m +O(\sigma^2+m)\right).
\end{aligned}
\end{equation}
\end{proposition}
\begin{proof}
As $N_\rho(k, k';0)=N(k, k';0) = \binom{m}{k}\mbf{1}\{k'=k\}$, we get
\[\sum_{k, k'=\ell-\sigma}^\ell \binom{m}{k}^{-1/2}\binom{m}{k'}^{-1/2}  N_\rho(k, k';0)  =\sum_{k, k'=\ell-\sigma}^\ell \binom{m}{k}^{-1/2}\binom{m}{k'}^{-1/2}  \binom{m}{k}\mbf{1}\{k'=k\}  = \sigma+1.\]
Moreover, by \cref{cor:triple-q}, we obtain
\[
N_\rho(k, k', 1;0) =(k+1)\binom{m}{k+1}\mbf 1\{k'=k+1\}+\beta k\binom{m}{k}\mbf{1}\{k'=k\}+(m-k+1)\binom{m}{k-1}\mbf{1}\{k'=k-1\},
\]
so we split into cases of $k'\in \{k-1, k, k+1\}$ to compute
\begin{align*}
& \sum_{k, k'=\ell-\sigma}^\ell \binom{m}{k}^{-1/2}\binom{m}{k'}^{-1/2}  N_\rho(k, k', 1;0) 
\\ & = 
\sum_{k=\ell-\sigma}^\ell 
\beta k
+   \sum_{k=\ell-\sigma}^{\ell-1} (k+1) \binom{m}{k+1}^{1/2}\binom{m}{k}^{-1/2}+\sum_{k=\ell-\sigma+1}^{\ell} (m-k+1) \binom{m}{k-1}^{1/2}\binom{m}{k}^{-1/2}
  \\ & =  \frac{1}{2}\beta (2\ell-\sigma)(\sigma+1)+ \sum_{k=\ell-\sigma}^{\ell-1} \sqrt{(k+1)(m-k)}+\sum_{k=\ell-\sigma+1}^{\ell} \sqrt{k(m-k+1)}
  \\ & = \frac{1}{2}\beta (2\ell-\sigma)(\sigma+1)+  2m\sum_{k=\ell-\sigma+1}^{\ell } \sqrt{\frac{k}{m}\left(1-\frac{k-1}{m}\right)}
\\ & =\beta (\sigma+1) \left(\mu+\delta -O\left(\frac{\sigma}{m}\right)\right)m+ 2\sigma m\left(\sqrt{(\mu+\delta)(1-\mu-\delta)}+O\left(\frac{\sigma}{m}\right)\right)
\\ & = \left[\beta  (\mu+\delta) + 2\sqrt{(\mu+\delta)(1-\mu-\delta)}\right]\left(\sigma m +O(\sigma^2+m)\right).
\end{align*}
Let $g(x)=\sqrt{x(1-x)}$, defined for $x\in [0, 1]$. In the second to last equality, there are $\sigma$-many summands, each of which is $g(\mu+\delta)+O\left({\sigma}/{m}\right)$ since $\ell/m = \mu+\delta+O(1/m)$.  Here, we are using the fact that for fixed $\mu$ and $\sigma/m \to 0$, we have $g((\mu + \delta + O(\sigma/m)) = g(\mu + \delta) + O(\sigma/m)$. 
\end{proof}
Recalling that $\mathbb{E}[q_0] = 1$, \Cref{prop:step-0} tells us that, \emph{if} we were to ignore the $t > 0$ terms in \cref{eq:master-t}, we \emph{would} have
\begin{equation}\label{eq:recover-biased-scl}
\begin{aligned}
\mb{E}_{\mb{P}_u}[s(x)] &\approx \rho+\sqrt{\rho(1-\rho)}\left[\beta  (\mu+\delta) + 2\sqrt{(\mu+\delta)(1-\mu-\delta)}\right]
\\ & =\rho+(1-2\rho)(\mu+\delta)+2\sqrt{\rho(1-\rho)(\mu+\delta)(1-\mu-\delta)}
\\ & = \rho(1-\mu-\delta)+(1-\rho)(\mu+\delta)+2\sqrt{\rho(1-\rho)(\mu+\delta)(1-\mu-\delta)}
\\ & = \on{SCL}_\rho(\mu+\delta).
\end{aligned}
\end{equation}
To prove this rigorously, we will show that the $t>0$ terms in \cref{eq:master-t} are asymptotically smaller than their $t=0$ counterparts. Via \cref{lem:key-fourier}, we have $\mathbb{E}[q_t(x)] = 0$ for $1 \leq t< d^\perp$, so it suffices to consider 
\begin{equation}\label{eq:int-s}
2\mu m \sim d^\perp \le t\le 2\ell \sim 2(\mu+\delta)m.
\end{equation}
For these values of $t$, we control both the sum of binomial coefficients in \cref{eq:master-t} as well as $\mb{E}[q_t(x)]$, to show that for any $t$ and $k, k'\in [\ell-\sigma, \ell]$ that under conditions of the main theorems, we have
\begin{equation}\label{eq:biased-goal}
\frac{\tilde{N}_\rho(k, k' ;t)}{\sqrt{\binom{m}{k}\binom{m}{k'}}}\cdot |\mb{E}[q_t]|\le e^{-\Omega(m)}\quad\text{where}\quad \tilde{N}_\rho(k, k';t) \coloneqq \sum_{j=0}^t|\beta|^{j}\binom{t}{j}\binom{t-j}{\frac{t+k-k'-j}{2}}\binom{m-t}{\frac{k+k'-t-j}{2}}.
\end{equation}
We outline a three-step plan to do so.
\begin{enumerate}

    \item First, we show that each of the $\sigma^2$ terms in the sum in \cref{eq:master-t} indexed by $(k, k')$ is asymptotically close to the leading one indexed by $k=k'=\ell$ (with possibly minus one from parity).
    \item Then, we compute the leading term of \cref{eq:master-t} where $k, k'= \ell$. For the balanced special case where $\rho\sim 1/2$, this turns out to be $\exp(mE(\mu, \delta)+o(m))$ where $E$ is defined in \cref{thm:main-avg}.
    \item Finally, we borrow and improve techniques from local leakage resilience (see \cref{sec:llr}) to control $|\mb{E}[q_t(x)]|$.  In that context, as discussed in \cref{sec:overview}, this quantity is the per-transcript control of the Fourier proxy in \cref{eq:BDIR}. This step is the content of \cref{sec:step-3}.
\end{enumerate}
For the remainder of this section, we execute the first two steps.
We begin with the first step.

\begin{proposition}\label{prop:step-1}
For any
$\rho\in (0, 1)$, $\mu \in (0, 1/2]$, $\delta \in [0, 1-\rho-\mu]$, $t\in [d^\perp, 2\ell]$, and $k, k'\in [\ell-\sigma , \ell]$, 
\begin{equation}\label{eq:biased-Nkk-1}
\frac{N_\rho(k, k'; t)}{\sqrt{\binom{m}{k}\binom{m}{k'}}}\le
e^{o(m)} \frac{\tilde{N}_\rho(\ell, \ell-\mbf{1}\{2\nmid t\} ;t)}{\binom{m}{\ell}}.
\end{equation}
\end{proposition}
\begin{proof}
We prove the case where $t$ is even. The odd $t$ case follows similarly. Observe that 
\[ 0<\mu\le \frac{\ell}{m}\le\mu+\delta \le 1-\rho,\]
so $\ell/m$ is bounded away from $0$ and $1$. Then, for any $k\in [\ell-\sigma, \ell]$, we have that
\begin{equation}\label{eq:binom-bound}
\frac{\binom{m}{\ell}}{\binom{m}{k}} = \prod_{j=0}^{\ell-k-1}\frac{m-\ell+j+1}{\ell-j} \le \left(\frac{m+\sigma}{\mu m -\sigma}\right)^\sigma \le C_0^\sigma,
\end{equation}
for sufficiently large $m$ since $\sigma\ll m$, where constant $C_0=C_0(\rho, \mu)>0$. This gives a bound on the ratio of denominator of \cref{eq:biased-Nkk-1}, i.e.
\begin{equation}
\binom{m}{k}^{-1/2}\binom{m}{k'}^{-1/2}\le C_0^{\sigma}\binom{m}{\ell}^{-1}.
\end{equation}

We similarly bound the ratio of numerators of \cref{eq:biased-Nkk-1}. Recall from \cref{eq:biased-goal} that
\[
N_\rho(k, k';t) = \sum_{j=0}^t\beta^{j}\binom{t}{j}\binom{t-j}{\frac{t+k-k'-j}{2}}\binom{m-t}{\frac{k+k'-t-j}{2}}.
\]
It suffices to show for every $j$ with a nonzero summand that the $j$-th summand in absolute value is at most $m^{O_{\rho,\mu}(\sigma)}$ times some corresponding term when $k=k'=\ell$ and $\beta$ replaced by its absolute value. Without loss of generality, we may assume that $j$ is even and that the corresponding summand for $k=k'=\ell$ is nonzero. Otherwise, we compare with the nearest feasible even integer $j'$. Since $|j-j'|=O(\sigma)$, this incurs at most an $m^{O(\sigma)}$ factor from the ratios of factorials and the factor $|\beta|^{{j-j'}}$.
That is, it remains to show that, 
\[ {|\beta|}^{j}\binom{t}{j}\binom{t-j}{\frac{t+k-k'-j}{2}}\binom{m-t}{\frac{k+k'-t-j}{2}} \leq {m^{O_{\rho, \mu}(\sigma)}}\cdot \left(|\beta|^{j}\binom{t}{j}\binom{t-j}{\frac{t-j}{2}}\binom{m-t}{\frac{2\ell-t-j}{2}} \right). \]
for all remaining $j$. The middle binomial coefficient ${t -j \choose \frac{t + k - k' - j}{2} }$ is maximized at the center $\binom{t-j}{(t-j)/2}$ with $k=k'$. For the last binomial coefficient ${m - t \choose \frac{k+k'-t-j}{2} }$, we imagine stepping from $\frac{2\ell - t -j}{2} = \ell - \frac{t+j}{2}$ to $\frac{k + k' - t - j}{2} = \frac{k+k'}{2} - \frac{t+j}{2}$ by iteratively reducing $\ell$ by one until we reach $\frac{k+k'}{2}$. Each time we reduce by one, since both endpoints are feasible, we gain at most a $O(m)$ factor. This is similar to the argument in \cref{eq:binom-bound}. We gain at most $\sigma$ such factors, which combining with $C_0^\sigma$ from the denominator {and a factor of at most $m+1$ from summing over $j$, gives a total factor of $m^{O_{\rho,\mu}(\sigma)}=e^{o(m)}$}, giving the desired bound.
\end{proof}
\begin{remark}\label{rem:tildeN}
We remark that using $\tilde{N}$ instead of $N$ is crucial to prevent any possible cancellations in the sum in $N$. However, for $\rho\ge 1/2$, this degrades our bound. In fact, our bound on $\mu_1$ \emph{decreases} for $\rho$ right after $1/2$.  As $\mu_1(\rho)$ is monotonic in $\rho$, the bound can be improved by replacing the red ``bump'' in \cref{fig:phase} with the dashed segment.   The absolute value around $\beta$ is responsible for this behavior in our bound.
\end{remark}
We observe that for the last binomial factor, if $\mu+\delta \le 1/2$ (e.g. for $\rho\ge 1/2$), then
$\ell - (t+j)/2\le \ell-t/2\le (m-t)/2$, so the left side is maximized when $k+k'$ is maximized, i.e. at $k=k'=\ell$.

We move onto step (2) and control \cref{eq:master-t} by considering the worst-case $s$ and $j$ for a fixed $\mu, \delta, \rho$ to obtain some numerical optimization program. Before giving a general bound with $E_\rho(\mu, \delta,\tau)$ for $\rho\in (0, 1)$, we study the balanced case where $\rho \sim 1/2$ where a further simplification is possible from $E_\rho(\mu, \delta, \tau)$ to $E(\mu, \delta)$ in \cref{thm:main-avg}.
\begin{lemma}\label{lem:step-2}
With $\rho\sim 1/2$, $\delta \in [0, 1/2-\mu]$, and $\ell = (\mu+\delta)m\le m/2$, for any $t\in [d^\perp, 2\ell]$ 
\begin{equation}
\frac{\tilde{N}_\rho(\ell,\ell-\mbf{1}\{2\nmid t\};t)}{\binom{m}{\ell}}\le e^{\left[E(\mu, \delta)+o(1)\right]m}.
\end{equation}
\end{lemma}
\begin{proof}
We prove the case where $t$ is even. The odd $t$ case follows similarly.
We claim that $
 \tilde{N}_\rho(\ell, \ell;t)\le  e^{o(m)} N(\ell,\ell;t)
$. 
To prove this, observe $|\beta| =o(1)$ and expand
\[\tilde{N}_\rho(\ell, \ell;t) = \sum_{j=0}^t|\beta|^{j}\binom{t}{j}\binom{t-j}{\frac{t-j}{2}}\binom{m-t}{\ell-\frac{t+j}{2}}.\]
Note that $N(\ell, \ell;t)$ is exactly the $j=0$ term.
The middle binomial coefficient is clearly decreasing in $j$, and so is the last binomial coefficient as $\ell - (t+j)/2\le \ell-t/2\le (m-t)/2$ as $\ell\le m/2$. 
Now, we apply standard binomial coefficient bounds and $x\le e^x$ to obtain
\[ |\beta|^j\binom{t}{j}\le \left(\frac{e|\beta|t}{j}\right)^j\le e^{e|\beta|t} = e^{o(m)},\]
where the last step holds since $t\lesssim m$ and $|\beta|\ll 1$ since $\rho \sim 1/2$. 
This implies that $\tilde{N}_\rho(\ell, \ell;t) \le e^{o(m)}N(\ell, \ell;t)$.
Now, we show that $N(\ell, \ell;t)$ is maximized for $t\ge d^\perp$ at $t=d^\perp$, and the maximum value is asymptotically $e^{mE(\mu, \delta)}$. First, we compute that
\begin{equation}\label{eq:Nkk-comp-1}
\frac{N(\ell, \ell;t)}{\binom{m}{\ell}}=\frac{\binom{t}{t/2}\binom{m-t}{\ell-t/2}}{\binom{m}{\ell}}=\frac{\binom{\ell}{t/2}\binom{m-\ell}{t/2}}{\binom{m}{t}}.
\end{equation}
where the second equality holds upon cross-multiplying and observing that both sides count the number ways to partition $[m]$ into sets of sizes $t/2, t/2, \ell-t/2, m-\ell-t/2$. 

Now, we apply \cref{fact:stirling} to \cref{eq:Nkk-comp-1} and compute for
$t \coloneqq 2(\mu+\tau)m$ where $0\le \tau \le \delta \le {1}/{2}-\mu $ that
 \begin{equation}\label{eq:Nkk-comp-2}
\frac{N(\ell, \ell;t)}{\binom{m}{\ell}}= m^{O(1)}\exp \left(m \left[(\mu+\delta)H\left(\frac{\mu+\tau}{\mu+\delta}\right)+(1-\mu-\delta)H\left(\frac{\mu+\tau}{1-\mu-\delta}\right)-H(2\mu +2\tau )\right]\right).
\end{equation}
If we differentiate the expression in the square-brackets with respect to $\tau$: as $H'(x)=\log (1/x-1)$,
\begin{align*}
    & \log \left(\frac{\mu+\delta}{\mu+t}-1\right)+\log \left(\frac{1-\mu-\delta}{\mu+t}-1\right)-2\log \left(\frac{1}{2\mu+2t}-1\right) 
    \\ & = \log \left[\left(\frac{\delta-\tau}{\mu+\tau}\right)\left( \frac{1-2\mu-\delta-\tau}{\mu+\tau}\right)\left(\frac{1-2\mu-2\tau}{2(\mu+\tau)}\right)^{-2}\right]
    \\ & = \log \left(\frac{4(\delta-\tau)(1-2\mu-\delta-\tau)}{(1-2\mu-2\tau)^2}\right)
    \\ & \le 0,
\end{align*}
where we use $4xy\le (x+y)^2$ with $x=\delta-\tau$ and $y=1-2\mu-\delta-\tau$. Therefore, \cref{eq:Nkk-comp-2} is maximized at $\tau=0$ which recovers the exponent $E(\mu, \delta)$ in \cref{eq:EFH}. The lemma follows.
\end{proof}
In general, for $\rho \in (0, 1)$ and $\delta >0$, the sum over $j$ in $\tilde{N}_\rho$ is not maximized at $j=0$ (or $j\ll m$) even up to $e^{o(m)}$ factors. It is a convex optimization problem to maximize over $j$ for the rate in the exponent in terms of $\rho, \mu, \delta$ and $s\coloneqq 2(\mu+\tau)m$ for $0\le \tau\le \min(\delta, 1/2-\mu)$, i.e.
\begin{equation}\label{eq:def-E-rho}
E_\rho(\mu, \delta, \tau)\coloneqq \lim_{m\to\infty}\frac{1}{m}\log\left(\frac{\tilde{N}_\rho(\ell,\ell-\mbf{1}\{2\nmid s\};2(\mu+\tau)m)}{\binom{m}{\ell}}\right).
\end{equation}
Since $\tilde{N}_\rho$ is a sum over $O(m)$ many terms indexed by $j$, we let $\gamma = j/m$ and extract the rate as
\begin{equation}
\begin{aligned}
E_\rho(\mu, \delta, \tau) & = 2(\mu+\tau)\log 2-H(\mu+\delta)
\\ & \quad +
\max_{\gamma} \left\{\gamma\log\left(\frac{|1-2\rho|}{2\sqrt{\rho(1-\rho)}}\right)+2(\mu+\tau)H\left(\frac{\gamma}{2(\mu+\tau)}\right)+(1-2\mu-2\tau) H\left(\frac{\delta-\tau-\gamma/2}{1-2\mu-2\tau}\right)\right\}
\end{aligned}
\end{equation}
maximized over $\gamma $ so that the binomial coefficients/entropy functions are well-defined. Note that $|\beta|=|1-2\rho|/\sqrt{\rho(1-\rho)}$ and we have $(2\mu+2\tau-\gamma)\log 2$ from $\binom{s-j}{(s-j)/2}$. 

In summary, combining \cref{eq:recover-biased-scl} and the first two steps allows us to reduce to the following.
\begin{lemma}\label{lem:step-1-2}
For any
$\rho\in (0, 1)$, $\mu \in (0, 1/2]$, and $\delta \in [0, 1-\rho-\mu]$, we have $\mb{E}_{\mb{P}_u}[s(x)]\ge \on{SCL}_\rho(\mu+\delta)$ provided for $t=2(\mu+\tau)m$ and sufficiently large $m$ that 
\begin{equation}
\sup_{\tau \in [0, \delta]} \left\{E_\rho (\mu, \delta, \tau)+\frac{1}{m}\log\left|\mb{E}[q_t(x)]\right|\right\} < 0.
\end{equation}
In the balanced case where $\rho \sim 1/2$, we replace $\delta \in [0, 1/2-\mu]$ and $E(\mu, \delta)$ for $E_\rho(\mu, \delta, \tau)$.
\end{lemma}
\begin{proof}
By \cref{lem:step-2} and \cref{eq:def-E-rho}, under assumptions in the lemma, we know that
\begin{equation}
	\frac{\tilde{N}_\rho(\ell, \ell-\mbf{1}\{2\nmid t\};t)}{\binom{m}{\ell}}\cdot |\mb{E}[q_t]|\le e^{-\Omega(m)}.
\end{equation}
Then, by \cref{prop:step-1} we know that there exists $c=c(\mu, \delta)$ such that for any $\tau \in [0, \delta]$ with $t=2(\mu+\tau)m$ and $k, k'\in [\ell-\sigma, \ell]$
\begin{equation}
	\frac{|N_\rho (k, k';t)|}{\sqrt{\binom{m}{k}\binom{m}{k'}}}\cdot |\mb{E}[q_t]|\le e^{-\Omega(m)}
\end{equation} 
and by \cref{cor:triple-q} the same holds replacing $N_\rho (k, k';t)$ by $N_\rho (k, k',1;t)$.
Recall \cref{eq:s-to-q1,eq:master-t}. Plugging in \cref{eq:biased-goal,prop:step-0}, we have
\begin{equation}
\mb{E}[s(x)]\ge \rho+\frac{\sqrt{\rho(1-\rho)}}{m}\cdot \frac{\left[\beta  (\mu+\delta) + 2\sqrt{(\mu+\delta)(1-\mu-\delta)}\right](\sigma m +O(\sigma^2+m))+m^{O(1)}e^{-cm}}{(\sigma+1)+m^{O(1)}e^{-cm}}
\end{equation}
This establishes the first step of \cref{eq:recover-biased-scl} rigorously. The other steps give $\mb{E}[s(x)]\ge \on{SCL}_\rho(\mu+\delta)$.
\end{proof}
\subsection{Controlling Expected Discrepancy}
\label{sec:step-3}
In this section, we bound $|\mb{E}[q_t]|$ when the rate $2\mu$ is large, using techniques from \cite{bdir,mpsw21,mnpw22}. Recall $\mb{E}[q_t(x)] = \sum_{y\in C_t^\perp} \prod_{i:y_i\ne 0}\wh{g}_i(y_i)$ from \cref{lem:key-fourier}. The starting point is to compute that
\begin{equation}\label{eq:biased-indicator-comp}
\mb{E}[q_t(x)] = \sum_{y\in C_t^\perp} \prod_{i:y_i\ne 0}\wh{g}_i(y_i)
 =\sum_{y\in C_t^\perp} \prod_{i:y_i\ne 0}\frac{\wh{\mbf{1}_{S_i}}(y_i)}{\sqrt{\rho(1-\rho)}}\prod_{i:y_i= 0}\frac{\wh{\mbf{1}_{S_i}}(y_i)}{\rho}
 = \rho^{t/2-m}(1-\rho)^{-t/2} \sum_{y\in C_t^\perp} \prod_{i=1}^m\wh{\mbf{1}_{S_i}}(y_i)
\end{equation}
where we recall the scaling of $g_i$ from \cref{eq:def-g}. In particular, in the balanced case where $\rho\sim 1/2$, we recover \cref{eq:balanced-llr} up to the $o(1)$ error, i.e.
\begin{equation}\label{eq:balanced-indicator-comp}
\left|\mb{E}[q_t(x)]\right|=  (2+o(1))^m \left|\sum_{y\in C^\perp_t} \prod_{i=1}^{m}\wh{\mbf{1}_{S_i}}(y_i)\right|.
\end{equation}

For secret sharing in \cref{sec:LRS_connection}, fixing the secret $s$ makes the shares uniform on a translate of the code obtained by shortening $\mathcal{C}$ at coordinate $0$. Applying the expansion in \cref{eq:BDIR} to the shortened code with the leakage sets translated by the negative shift gives the distance $d_{\mathrm{TV}}(L(s),L')$. By the triangle inequality, a bound of $2^{-\Omega(m)}$ uniform in $s$ gives one-bit local-leakage resilience.

The key tool to control these Fourier coefficients is the following standard fact used by works on local leakage resilience~\cite{bdir,mpsw21,mnpw22}. It is essentially saying the maximizer $S$ is an interval of length $\rho p$.  We include a proof for completeness.
\begin{fact}[\cite{bdir}]
\label{fact:arc-max}
 For any subset $S\subset\mb{F}_p$ of size $|S|=\rho p$, $\left|\wh{\mbf{1}_S}(0)\right|=\rho$ and for $z\ne 0$, as $p\to\infty$
\begin{equation}
\left|\wh{\mbf{1}_S}(z)\right|\le \frac{|\sin(\rho\pi)|}{\pi}+O\left(\frac{1}{p^2}\right).
\end{equation}
\end{fact}
\begin{proof}
Let $s=\rho p$ and $\omega = \exp(2\pi i /p)$. The $z=0$ case is clear. Multiplying by $z\ne 0$ permutes $\mb{F}_p$, so
\begin{equation}
   p\max_{|S|=\rho p} \left|\wh{\mbf{1}_S}(z)\right| =    \max_{|S|=\rho p}\left|\sum_{x\in S}\exp\left(\frac{2\pi i xz}{p}\right)\right|= \max_{|S|=\rho p}\left|\sum_{x\in S}\omega^x\right|=\max_{|S|=\rho p} \sum_{x\in S}\cos\left(\frac{2\pi x}{p}-\theta\right),
\end{equation}
for some argument $\theta$. Since $\cos$ is decreasing as the angular distance from $\theta$ increases on $[0,\pi]$, it is maximized by choosing the $\rho p$ points among the $p$-th roots of unity whose arguments are closest to $\theta$, i.e. $\rho p$ consecutive roots. By rotational invariance, we may therefore take $\{0, 1, \ldots, s-1\}$, so
\[
 \left|\wh{\mbf{1}_S}(z)\right| \le \frac{1}{p}
\left|\sum_{j=0}^{s-1}\omega^j\right|
=
\frac{|1-\omega^s|}{p|1-\omega|}
=
\frac{2|\sin(\pi s/p)|}{2p\sin(\pi/p)}
=
\frac{|\sin(\rho\pi)|}{p\sin(\pi/p)}
\le \frac{|\sin(\rho\pi)|}{\pi}+O\left(\frac{1}{p^2}\right),
\]
for $z\ne 0$, since
$
\sin(\pi/p)={\pi}/{p}+O\left({1}/{p^3}\right)$ as $p\to\infty$.
\end{proof}
\begin{remark}\label{rem:p_prime}
We remark that it is crucial that $p$ is prime: otherwise there may be other Fourier coefficients as large as $|S|/p$. This is the reason we cannot directly work over extension fields. The same is true for the methods of \cite{bdir,mnpw22,mpsw21} to bound the Fourier proxy.
\end{remark}
Now, we build on techniques of local leakage resilience works to obtain the following theorem.
\begin{theorem}
\label{thm:step-3}
Fix any $\mu, \delta$ satisfying \cref{eq:def-delta}. Fix $t\in[d^\perp,2\ell]$. Suppose there exist sets $B_1, \ldots, B_J\in \binom{[m]}{2n-m}$ and $\lambda \in [0, 1]$ such that for every $D\in\binom{[m]}{t}$, there exists $j\in [J]$ with $|D\cap B_j|\ge \lambda m$. Then
\begin{equation}
\left|\mb{E}[q_t(x)]\right| \lesssim J\rho^{n-m}\left(\frac{\rho}{1-\rho}\right)^{t/2}\left|\frac{\sin(\rho\pi)}{\rho\pi}\right|^{\lambda m}.
\end{equation}
\end{theorem}
\begin{proof}
By assumption, we know that for every $y\in C^\perp_{t}$, there exists $j\in [J]$ such that $|\on{supp}y\cap B_j|\ge \lambda m$. We can therefore partition $C_t^\perp = \bigcup_{j=1}^J Y_j$ where $Y_j$ are pairwise disjoint and
$|\on{supp}y\cap B_j|\ge \lambda m$ for every $y\in Y_j$.
Then, we split the sum over $y$ into these $Y_j$: for each $j$, we pick out coordinates indexed by $B_j$ (of size $m-2k=2n-m$), and split the other coordinates in $[m]$ arbitrarily into two sets $L_j$ and $R_j$ of size $k=m-n$ each. By Cauchy-Schwarz, 
\begin{equation}
\begin{aligned}
\left|\sum_{y\in C^\perp_t} \prod_{i=1}^{m}\wh{\mbf{1}_{S_i}}(y_i)\right|  & \le \sum_{j=1}^J \left|\sum_{y\in Y_j} \prod_{i=1}^{m}\wh{\mbf{1}_{S_i}}(y_i)\right| 
\\ & \le \sum_{j=1}^J  \left(\sum_{y\in Y_j} \prod_{i\in L_j}\left|\wh{\mbf{1}_{S_i}}(y_i)\right|^2\right)^{1/2}\left(\sum_{y\in Y_j} \prod_{i\in R_j}\left|\wh{\mbf{1}_{S_i}}(y_i)\right|^2\right)^{1/2}\max_{y\in Y_j} \prod_{i\in B_j}\left|\wh{\mbf{1}_{S_i}}(y_i)\right|
\\ & \le \sum_{j=1}^J  \left(\sum_{y\in C^\perp} \prod_{i\in L_j}\left|\wh{\mbf{1}_{S_i}}(y_i)\right|^2\right)^{1/2}\left(\sum_{y\in C^\perp} \prod_{i\in R_j}\left|\wh{\mbf{1}_{S_i}}(y_i)\right|^2\right)^{1/2}\max_{y\in Y_j} \prod_{i\in B_j}\left|\wh{\mbf{1}_{S_i}}(y_i)\right|
\\ & = \sum_{j=1}^J\left(\prod_{i\in L_j}\left[\sum_{y_i\in \mb{F}_p}\left|\wh{\mbf{1}_{S_i}}(y_i)\right|^2\right]\right)^{1/2}\left(\prod_{i\in R_j}\left[\sum_{y_i\in \mb{F}_p}\left|\wh{\mbf{1}_{S_i}}(y_i)\right|^2\right]\right)^{1/2}\max_{y\in Y_j} \prod_{i\in B_j}\left|\wh{\mbf{1}_{S_i}}(y_i)\right|
\\ & = \sum_{j=1}^J\left(\prod_{i\in [m]\setminus B_j} \Vert \wh{\mbf{1}_{S_i}}\Vert_2\right)\cdot \max_{y\in Y_j} \prod_{i\in B_j}\left|\wh{\mbf{1}_{S_i}}(y_i)\right|
\\ & \le \rho^{k} \sum_{j=1}^J \max_{y\in Y_j}\left|\frac{\sin(\rho\pi)}{\pi}+O\left(\frac{1}{p}\right)\right|^{|\on{supp} y \cap B_j|}\rho^{|B_j\setminus \on{supp} y |}
\\ & \lesssim \rho^{k} \sum_{j=1}^J \max_{y\in Y_j}\left|\frac{\sin(\rho\pi)}{\rho\pi}\right|^{|\on{supp} y \cap B_j|}\rho^{|B_j\setminus \on{supp} y |+|B_j\cap \on{supp} y |}
\\ &  = \rho^{k+(m-2k)} \sum_{j=1}^J\max_{y\in Y_j}  \left|\frac{\sin(\rho\pi)}{\rho\pi}\right|^{|\on{supp} y \cap B_j|} 
\\ & \le \rho^{n} J\left|\frac{\sin(\rho\pi)}{\rho\pi}\right|^{\lambda m},
\end{aligned}
\end{equation}
where in the fourth line we note that the projection maps are bijections from $C^\perp$ to $\mb{F}_p^{L_j}$ and $\mb{F}_p^{R_j}$ by the MDS property of $C^\perp$, and recall Parseval's identity and that
\begin{equation}
\Vert \wh{\mbf{1}_{S_i}}\Vert_2^2  =\Vert \mbf{1}_{S_i}\Vert_2^2 = \mb{E}_{a\in \mb{F}_p}\left[\mbf{1}_{S_i}(a)^2\right] = \rho.
\end{equation}
In the seventh line, we have used the fact that
\[ \left| \frac{ \sin(\rho \pi)}{\pi} + O\left(\frac{1}{p}\right) \right|^{|\mathrm{supp} y \cap B_j|} =\left| \frac{ \sin(\rho \pi)}{\pi}\right|^{|\mathrm{supp} y \cap B_j|} \left(1 + O\left( \frac{1}{p}\right) \right)^{|\mathrm{supp}(y) \cap B_j|} \lesssim \left| \frac{ \sin(\rho \pi)}{\pi}\right|^{|\mathrm{supp}(y) \cap B_j|}, \]
using the fact that $\rho$ is a constant bounded away from $1$, and that, by \cref{lem:mds-bound}, we have $p \gtrsim m \geq |\mathrm{supp}(y) \cap B_j|$
Now, the theorem follows from \cref{eq:biased-indicator-comp}.
\end{proof}
A first application of this theorem is obtained by choosing $J=1$ and $B_1$ to be any set. This is essentially the approach of \cite{bdir,mpsw21,mnpw22}.
Combined with \cref{lem:step-1-2}, we obtain an analog of \cref{thm:main-avg,thm:main-best} with worse rates: it is given as the green curve in \cref{fig:balanced}. 
\begin{theorem}[Off-the-shelf improvement in the balanced case]\label{thm:main-green}
Let $2\mu \in [0,1]$, and let $n,m$ be sufficiently large so that $n/m = 2\mu$.  Let $p$ be prime.  Then for any input sets $S_i$ of size $|S_i|/p \sim 1/2$, the MDS Max-LINSAT problem admits a solution $x$ with satisfaction ratio
\[
 s(x)\ge \on{SCL}_{1/2}(\mu+\delta)-o(1)
\]
for any $\delta \in [0, 1/2-\mu]$ such that $E(\mu, \delta)+F(\mu)< 0$ where
\begin{align*}
E(\mu, \delta) & \coloneqq(\mu+\delta)H\left(\frac{\mu}{\mu+\delta}\right)+(1-\mu-\delta)H\left(\frac{\mu}{1-\mu-\delta}\right)-H(2\mu)\\
F(\mu) & \coloneqq (1-2\mu)\log 2+(6\mu-2)\log\left(\frac{2}{\pi}\right)
\end{align*}   
\end{theorem}
\begin{proof}
In \cref{thm:step-3}, let $J=1$ and $B_1$ be any set of size $2n-m$. Then, for any $D\in\binom{[m]}{t}$
\[ |D\cap B_1|\ge |D|+|B_1|-m = t +(2n-m)-m\]
by inclusion--exclusion, so $\lambda = 6\mu-2+2\tau$. Then, by \cref{thm:step-3}, for any $\tau\in [0, \delta]$ and $t=2(\mu+\tau)m$ 
\[
\frac{1}{m}\log |\mb{E}[q_t(x)]|\le F(\mu)+o(1)
\]
where the $o(1)$ comes from $\rho = 1/2+o(1)$, so upon exponentiation of $\rho$ and $\rho/(1-\rho)$ to $\Theta(m)$, we have errors of $(1+o(1))^{\Theta(m)} =e^{o(m)}$. Now, the conclusion follows \cref{lem:step-1-2} and that $\on{SCL}_\rho(\mu+\delta) = \on{SCL}_{1/2}(\mu+\delta)+o(1)$ for $\rho \sim 1/2$ by continuity of the semicircle law function in $\rho$.
\end{proof}
\begin{corollary}
Let $2\mu \geq 0.78$, and let $p$ be prime.  Then for any sets $S_i \subseteq \mathbb{F}_p$ of size $|S_i|/p \sim 1/2$, there exists an asymptotically perfect solution to the MDS Max-LINSAT problem with inputs $S_i$.
\end{corollary}
\begin{proof}
Let $\delta =1/2-\mu$. We see that $E(\mu, 1/2-\mu)=0$ and observe $F(\mu)<0$ whenever $\mu>0.39$.
\end{proof}

Note that $0.78$ is exactly the rate \cite{bdir,mpsw21,mnpw22} obtain on local leakage resilience using the Fourier proxy \cref{eq:BDIR}. Indeed, for $\delta=1/2-\mu$ so $E(\mu, \delta)=0$, we are essentially controlling \cref{eq:BDIR}.
\section{Proof of Main Results}\label{sec:main-proofs}
In this section, we prove our main results, first for the balanced case when $\rho\sim 1/2$, then the general case for any $\rho\in (0, 1)$. The main idea to improve from \cref{thm:main-green} and \cite{bdir,mpsw21,mnpw22} is to optimize the application of \cref{thm:step-3}.
\subsection{The Balanced Case}
\label{sec:balanced-proofs}
We focus on the balanced case where $\rho\sim 1/2$.
To obtain \cref{thm:main-avg}, we choose the buckets $B_j$ so that we can guarantee that the intersection has density at least that of $D$, i.e. $t/m$. This extra exponential decay overwhelms polynomial cost $J$.
\begin{lemma}\label{lem:avg-buckets}
For any $\mu, \delta$ satisfying \cref{eq:def-delta}, and $t=2(\mu+\tau)m$ with $\tau\in [0, \delta]$
\begin{equation}
\left|\mb{E}[q_t(x)]\right| \lesssim m\rho^{n-m}\left(\frac{\rho}{1-\rho}\right)^{t/2}\left|\frac{\sin(\rho\pi)}{\rho\pi}\right|^{2(\mu+\tau)(4\mu-1)m}.
\end{equation}
\end{lemma}

\begin{proof}
We apply \cref{thm:step-3} with $J=m$ and $B_j=\{j, j+1, \ldots, j+(2n-m)-1 \}$ for $j\in [m]$, with indices taken modulo $m$. As the sets $B_j$ cover each element of $[m]$ exactly $2n-m$ times, we have that for any $D\in\binom{[m]}{t}$
\begin{equation}
\mb{E}_{j\in [J]} \left|D\cap B_j\right|= \frac{|D|(2n-m)}{m} = 2(\mu+\tau)(4\mu-1)m,
\end{equation}
so there exists some $j$ satisfying the condition of \cref{thm:step-3}, and the lemma follows.
\end{proof}
Now, we can prove \cref{thm:main-avg}.  We restate the theorem for the reader's convenience.
\thmA*

\begin{proof}[Proof of \cref{thm:main-avg}]
By \cref{lem:avg-buckets}, for any $\tau\in [0, \delta]$ and $t=2(\mu+\tau)m$ 
\[
\frac{1}{m}\log |\mb{E}[q_t(x)]|\le F(\mu)+o(1),
\]
where the $o(1)$ comes from $\rho = 1/2+o(1)$, so upon exponentiation of $\rho$ and $\rho/(1-\rho)$ to $\Theta(m)$, we have errors of $(1+o(1))^{\Theta(m)} =e^{o(m)}$. Now, the conclusion follows \cref{lem:step-1-2} and that $\on{SCL}_\rho(\mu+\delta) = \on{SCL}_{1/2}(\mu+\delta)+o(1)$ for $\rho \sim 1/2$ by continuity of the semicircle law function in $\rho$.
\end{proof}

We further improve \cref{thm:main-best} with exponentially many buckets $B_j$ to get intersection size of $\lambda m$. \cref{lem:avg-buckets} corresponds to choosing $\lambda = 2\mu(4\mu-1)$ and $J=m$ for worst-case $t=d^\perp$, and we now balance the tradeoff of increasing $\lambda$ at the cost of exponentially large $J$. Recall from \cref{eq:EFG} that
\[
G(\mu, \lambda ) \coloneqq \left(1-2\mu\right)\log 2 +H\left(2\mu \right) -\left(4\mu -1\right)H\left(\frac{\lambda }{4\mu -1}\right)-\left(2-4\mu \right)H\left(\frac{2\mu -\lambda }{2-4\mu }\right)+\lambda \log\left(\frac{2}{\pi}\right).
\]
\begin{lemma}\label{lem:best-buckets}
For any $\mu, \delta$ satisfying \cref{eq:def-delta}, balanced $\rho\sim 1/2$, $\varepsilon = \varepsilon(\mu)>0$, and $t\in [d^\perp, 2\ell]$
\begin{equation}
    \left|\mb{E}_{x\in \mb{F}_p^n} \left[q_t(x) \right]\right|\le e^{m[G(\mu, \lambda)+\varepsilon+o(1)]}.
\end{equation}
\end{lemma}
\begin{proof}
We apply \cref{thm:step-3}, and show the existence of the appropriate buckets $B_j$ from which the lemma would follow. To check the condition in \cref{thm:step-3}, it suffices to check it for sets $D$ of size exactly $d^\perp$.

For any $B_j$ of size $2n-m$, the number $N(\lambda)$ of sets $D\in\binom{[m]}{d^\perp}$ such that $|B_j\cap D|\ge \lambda m$ is 
\begin{equation}\label{eq:N-lambda}
N(\lambda) = \sum_{k\ge \lambda m} \binom{2n-m}{k}\binom{2(m-n)}{n+1-k},
\end{equation}
as $d^\perp=n+1$. Then, for uniformly random $2n-m$ element subset $B_j$ of $[m]$ and any $D\in\binom{[m]}{n+1}$,
\begin{equation}
\gamma\coloneqq \mb{P}\left(\left|D\cap B_j\right|\ge \lambda m\right)=\frac{N(\lambda)}{\binom{m}{n+1}} \gtrsim m^{-1/2}\exp\left(m\left[ (4\mu-1) H\left(\frac{\lambda}{4\mu-1}\right)+(2-4\mu)H\left(\frac{2\mu - \lambda}{2-4\mu}\right)-H(2\mu)\right]\right),
\end{equation}
where we pick out the $k=\lambda m$ summand only from \cref{eq:N-lambda}. Now, we choose $B_1, \ldots, B_J$ independently uniformly at random, so by the union bound, there exists $D$ whose intersection with every $B_j$ is less than $\lambda m$ with probability at most
\begin{equation}
\binom{m}{d^\perp}\left(1-\gamma\right)^J < \binom{m}{n+1}e^{-\gamma J}\lesssim m^{O(1)} \exp\left(mH(2\mu)-\gamma J\right) \ll 1,
\end{equation}
provided we choose any $\varepsilon =\varepsilon(\mu)>0$ and choose $J$ to be 
\begin{equation}
J= \exp\left(m\left[ -(4\mu-1) H\left(\frac{\lambda}{4\mu-1}\right)-(2-4\mu)H\left(\frac{2\mu - \lambda}{2-4\mu}\right)+H(2\mu)+\varepsilon\right]\right)
\end{equation}
With this choice of $J$, with positive probability
every $D$ has intersection at least $\lambda m$ with some $B_j$.  In particular, buckets $B_j$ satisfying this property exist.  Then, by \cref{thm:step-3} and \cref{eq:balanced-indicator-comp}
\begin{equation}
    \left|\mb{E}_{x\in \mb{F}_p^n} \left[q_t(x) \right]\right|\le  (2+o(1))^m \left|\sum_{y\in C^\perp_t} \prod_{i=1}^{m}\wh{\mbf{1}_{S_i}}(y_i)\right| \lesssim e^{o(m)}2^{m-n}\left(\frac{2}{\pi}\right)^{\lambda m}J
    \le e^{m[G(\mu, \lambda)+o(1)+\varepsilon]}.
\end{equation}
where $G$ is defined in \cref{eq:EFG}.
\end{proof}

Finally, we can prove \cref{thm:main-best}.  We restate the theorem for the reader's convenience.
\thmB*
\begin{proof}[Proof of \cref{thm:main-best}]
Fix any $\delta,\mu,\lambda$ such that $E(\mu, \delta)+G(\mu, \lambda)<0$ under assumptions of \cref{thm:main-best}, we choose $\varepsilon \coloneqq-(E(\mu, \delta)+G(\mu, \lambda ))/2$. By \cref{lem:best-buckets}, for any $\tau\in [0, \delta]$ and $t=2(\mu+\tau)m$ 
\[
\frac{1}{m}\log |\mb{E}[q_t(x)]|\le G(\mu, \lambda)+\varepsilon+o(1),
\]
where the $o(1)$ comes from $\rho = 1/2+o(1)$, so upon exponentiation of $\rho$ and $\rho/(1-\rho)$ to $\Theta(m)$, we have errors of $(1+o(1))^{\Theta(m)} =e^{o(m)}$. Now, the conclusion follows from \cref{lem:step-1-2} as
\begin{equation}
E(\mu, \delta)+\frac{1}{m}\log |\mb{E}[q_t(x)]|\le E(\mu, \delta)+G(\mu, \lambda)+\varepsilon+o(1)\le -\varepsilon+o(1).
\end{equation}
Finally, note that $\on{SCL}_\rho(\mu+\delta) = \on{SCL}_{1/2}(\mu+\delta)+o(1)$ for $\rho \sim 1/2$ by continuity of the semicircle law function in $\rho$. 
Now, we minimize $G(\mu, \lambda)$ in $\lambda$. Recall $H'(x)=\log(1/x-1)$ from \cref{fact:stirling} to obtain
\begin{equation}
0=\frac{\partial G}{\partial\lambda}(\mu, \lambda) = -\log \left(\frac{4\mu-1}{\lambda}-1\right)+\log \left(\frac{1-(4\mu-1)}{2\mu-\lambda}-1\right)+\log\left(\frac{2}{\pi}\right).
\end{equation}
Rearranging, we obtain the following quadratic equation in $\lambda$:
\begin{equation}
1=\frac{2}{\pi}\cdot\frac{\lambda}{4\mu-1-\lambda}\cdot\frac{1-(4\mu-1)-(2\mu-\lambda)}{2\mu-\lambda}.
\end{equation}
We solve for the larger root $\lambda\ge 2\mu(4\mu-1)$ to get
\begin{equation}
\lambda_\star\coloneqq \frac{A-\sqrt{A^2-8(1-2/\pi)\mu(4\mu-1)}}{2(1-2/\pi)}\quad\text{where}\quad A\coloneqq \frac{2}{\pi}+\left(1-\frac{2}{\pi}\right)(6\mu-1).
\end{equation}
This proves \cref{thm:main-best}.
\end{proof}

\subsection{The General Case}\label{sec:general-proofs}
In this section, we prove \cref{thm:main-biased} for arbitrary $\rho\in (0, 1)$.  Our proof is in analogy with \cref{thm:main-avg} using \cref{lem:avg-buckets}; for simplicity, we do not use the stronger control \cref{lem:best-buckets} and \cref{thm:main-best}. 
We restate the theorem for the reader's convenience.
\thmC*

Towards proving \cref{thm:main-biased}, define \begin{equation}
F_\rho \left(\mu, \delta, \tau\right)\coloneqq (2\mu-1)\log\rho +(\mu+\tau)\log \left(\frac{\rho}{1-\rho}\right)+2(\mu+\tau)(4\mu-1)\log \left(\frac{|\sin(\rho \pi)|}{\rho \pi}\right).
\end{equation}
Then, \cref{lem:avg-buckets} implies that for $t=2(\mu+\tau)m$ with $\tau\in [0, \delta]$, $ \frac{1}{m}\log |\mb{E}[q_t(x)]|\le F_\rho(\mu, \delta, \tau)+o(1)$.
Together with \cref{lem:step-1-2}, we deduce the following theorem, which will lead to \cref{thm:main-biased}.
\begin{theorem}\label{thm:unset-tau-biased}
Let $2\mu \in [0,1]$, and let $n,m$ be sufficiently large, with $n/m = 2\mu$.  Let $\rho \in (0,1)$. Let $p$ be prime, and fix any MDS generator matrix $B \in \mathbb{F}_p^{m \times n}$.  Then for any input sets $S_i$ with size $|S_i| = \rho p$, the MDS Max-LINSAT problem (with respect to $B$) admits a solution $x \in \mathbb{F}_p^n$ with satisfaction ratio
\begin{equation}
 s(x)\ge \on{SCL}_\rho(\mu+\delta)-o(1),
\end{equation}
for any $\delta \in [0, 1-\rho-\mu]$ such that
\begin{equation}
\sup_{0\le\tau\le \delta} \left\{E_\rho(\mu, \delta, \tau)+F_\rho(\mu, \delta, \tau)\right\} <0,
\end{equation}
where for binary entropy $H$ with natural logarithm base and
\begin{equation}
\begin{aligned}
E_\rho(\mu, \delta, \tau) & \coloneqq 2(\mu+\tau)\log 2-H(\mu+\delta)
\\ & \quad +
\max_{\gamma} \left\{\gamma\log\left(\frac{|1-2\rho|}{2\sqrt{\rho(1-\rho)}}\right)+2(\mu+\tau)H\left(\frac{\gamma}{2(\mu+\tau)}\right)+(1-2\mu-2\tau) H\left(\frac{\delta-\tau-\gamma/2}{1-2\mu-2\tau}\right)\right\}
\\ F_\rho \left(\mu, \delta, \tau\right) & \coloneqq  (2\mu-1)\log\rho +(\mu+\tau)\log \left(\frac{\rho}{1-\rho}\right)+2(\mu+\tau)(4\mu-1)\log \left(\frac{|\sin(\rho \pi)|}{\rho \pi}\right).
\end{aligned}
\end{equation}
\end{theorem}

\cref{thm:unset-tau-biased} has many parameters, and one must solve an optimization problem in order to understand its conclusions. For intuition, 
in \cref{fig:biased}, we include  two examples of optimization results, for $\rho\in\{0.4, 0.6\}$, showing both the optimal satisfaction ratio $s(x)$, and also the optimal choice of the parameters $\delta, \tau$, and $\gamma$, which we refer to as $\delta_{\star}, \tau_{\star}$, and $\gamma_{\star}$. In both cases, we notice that $\tau_\star=0$ so $s=d^\perp$ but $\gamma_\star$ is nonzero. Indeed, $\gamma_{\star}(\mu, \delta_\star, 0)$ increases until the threshold $\delta_\star= 1-\mu-\rho$.

For the remainder of this section, we prove that the maximizing $\tau \in [0, \delta]$ of $E_\rho+F_\rho $ is given by $\tau_\star = 0$. Then, \cref{thm:unset-tau-biased} recovers \cref{thm:main-biased}. This simplification is why we use the quantitatively weaker analog of \cref{thm:main-avg}, rather than \cref{thm:main-best}: If we replace $F_\rho $ with the analog of the best bucket control \cref{lem:best-buckets} for general $\rho$, it is not clear whether $\tau=0$ is the maximizer, and the resulting theorem would become even more complicated.

\begin{figure}[H]
    \centering
 \includegraphics[width=0.85\textwidth]{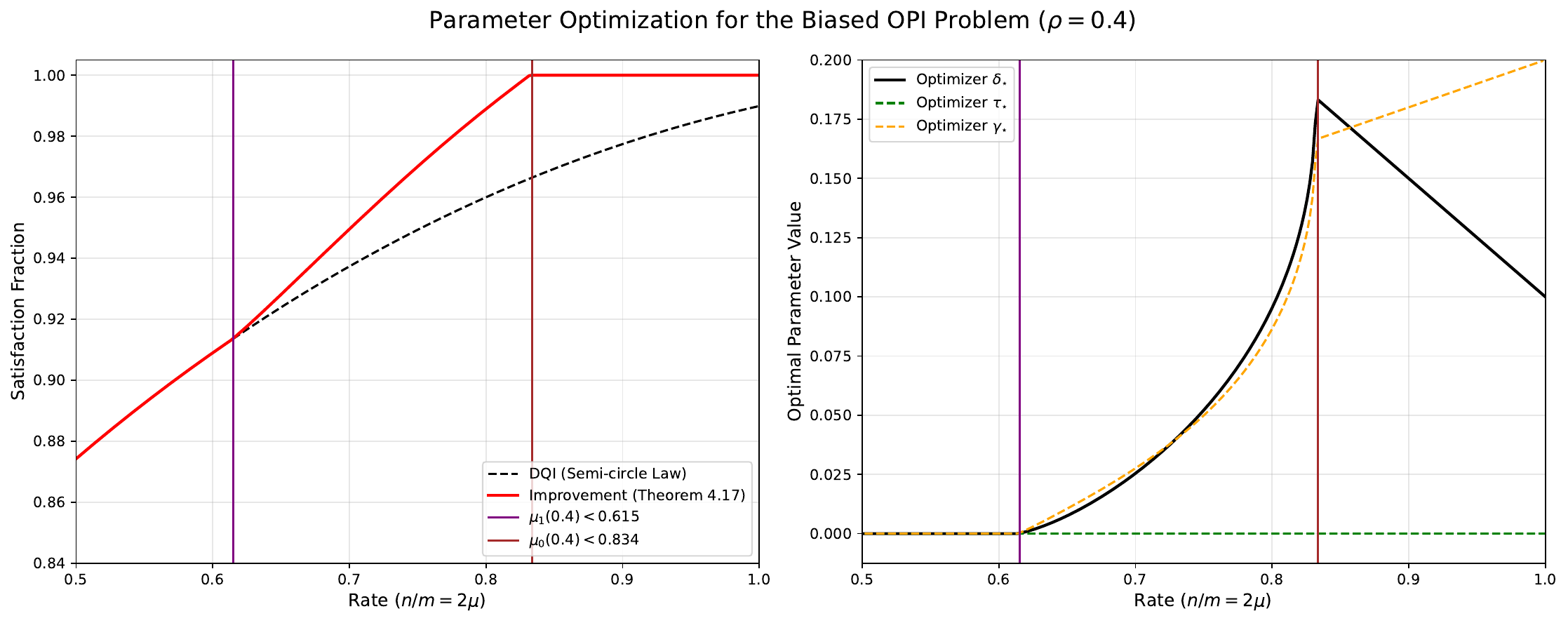}
 \includegraphics[width=0.85\textwidth]{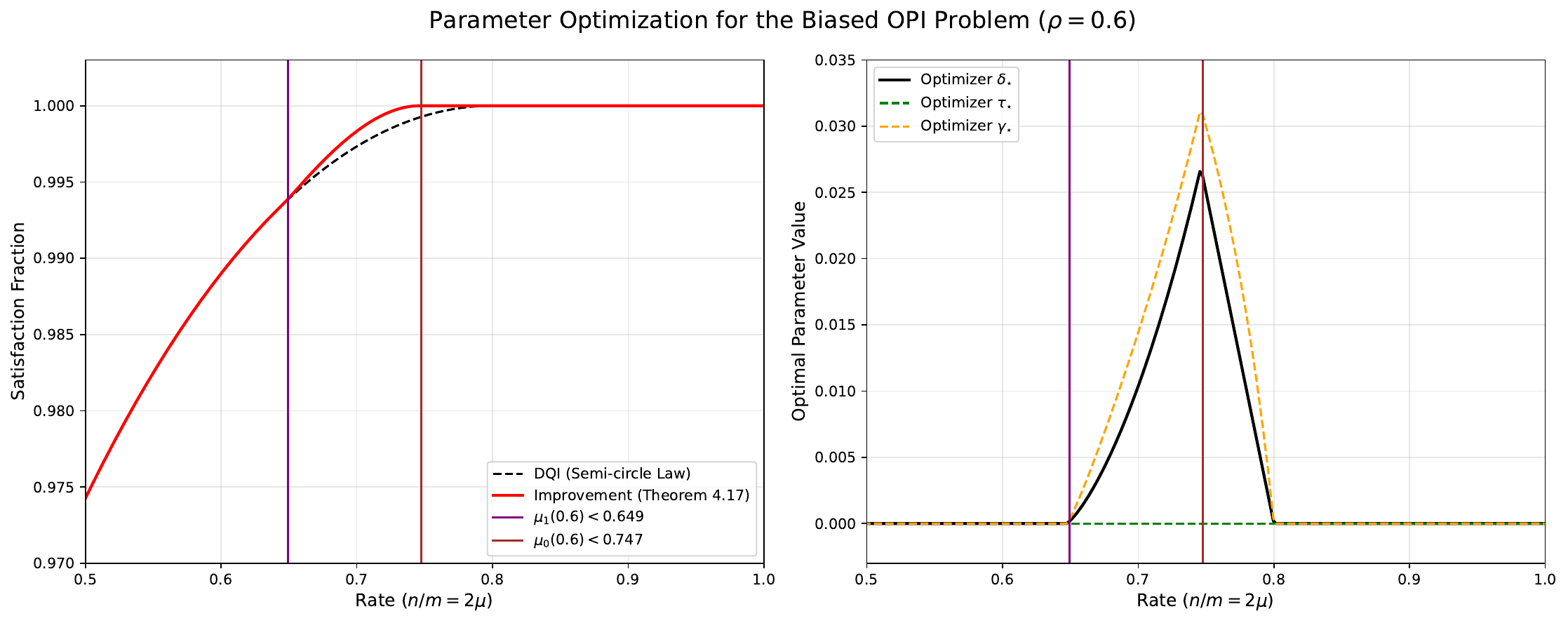}
    \caption{For two example values $\rho\in\{0.4, 0.6\}$, we plot values of optimizers $\tau_\star, \delta_\star, \gamma_\star$ from \cref{thm:unset-tau-biased}, and satisfaction fraction improvement from $\on{SCL}_\rho(\mu)$ as functions of rate $n/m = 2\mu\in [0,1]$. The left plot is analogous to \cref{fig:balanced}; in the right plot we observe $\tau_\star=0$ for all $\mu$.}
    \label{fig:biased}
\end{figure}

\begin{proof}[Proof of \cref{thm:main-biased}]  
Fix any $\mu, \delta$ satisfying \cref{eq:def-delta}. Denote the arguments of $H$ in $E_\rho$ by
\begin{equation}\label{eq:def-x-y}
x = x(\tau, \gamma) \coloneqq \frac{\delta-\tau-\gamma/2}{1-2(\mu+\tau)}\qquad\text{and}\qquad y = y(\tau, \gamma) \coloneqq \frac{\gamma}{2(\mu+\tau)}.
\end{equation}
Define $\mc{E}(\tau, \gamma)$ as follows for $ \beta = {(1-2\rho)}/{\sqrt{\rho(1-\rho)}}$, so that $E_\rho(\mu, \delta, \tau) = \max_{\gamma}\mc{E}(\tau, \gamma)$:
\[
\mc{E}(\tau, \gamma)\coloneqq 2(\mu+\tau)\log 2-H(\mu+\delta)
+\gamma\log\left(\frac{|\beta|}{2}\right)+2(\mu+\tau)H\left(y\right)+(1-2\mu-2\tau) H\left(x\right).
\]

Then, the maximizer $\gamma_\star(\tau)$ in the definition of $E_\rho$ is defined as the root of
\begin{align*}
0  = \frac{d\mc{E}}{d\gamma}(\tau, \gamma_\star(\tau))
 & = \frac{d}{d\gamma}\left[\gamma\log\left(\frac{|\beta|}{2}\right)+2(\mu+\tau)H\left(\frac{\gamma}{2(\mu+\tau)}\right)+(1-2\mu-2\tau) H\left(\frac{\delta-\tau-\gamma/2}{1-2\mu-2\tau}\right)\right]
\\ & = \log\left(\frac{|\beta|}{2}\right)+\log\left(\frac{1}{y}-1\right)-\frac{1}{2}\log\left(\frac{1}{x}-1\right)
\end{align*}
by \cref{fact:stirling}. Hence, at $\gamma=\gamma_\star(\tau)$, we have
\begin{equation}\label{eq:x-y-relation}
1=\frac{|\beta|}{2}\cdot \frac{1-y}{y}\sqrt{\frac{x}{1-x}}.
\end{equation}
Moreover, we compute
\begin{equation}
\frac{dF_\rho}{d\tau}= \log\left(\frac{\rho}{1-\rho}\right) +2(4\mu-1)\log\left(\frac{|\sin(\rho\pi)|}{\rho\pi}\right),
\end{equation}
and by optimality of $\gamma$ we have that
\begin{equation}
\frac{d\mc{E}}{d\gamma}\bigg\vert_{\gamma=\gamma_\star(\tau)}=0\implies 
\frac{d}{d\tau}\left[\mc{E}(\tau, \gamma_\star(\tau))\right] = \left[\frac{d\mc{E}}{d\tau} + \frac{d\mc{E}}{d\gamma}\frac{d\gamma}{d\tau}\right]_{\gamma=\gamma_\star(\tau)}  = \frac{d\mc{E}}{d\tau}\bigg\vert_{\gamma=\gamma_\star(\tau)}.
\end{equation}
Now, we compute via \cref{fact:stirling} that
\begin{align*}
\frac{d}{d\tau} [2(\mu+\tau)H(y)] & = 2H(y)+2(\mu+\tau)H'(y)\frac{dy}{d\tau}
\\ & = 2H(y)-2y\log\left(\frac{1-y}{y}\right)
\\ & = -2\log(1-y).
\end{align*}
Similarly, we have
\begin{align*}
\frac{d}{d\tau} [(1-2\mu-2\tau)H(x)] & = -2H(x)+(1-2\mu-2\tau)H'(x)\frac{dx}{d\tau}
\\ & = -2H(x)+(1-2\mu-2\tau)\log\left(\frac{1-x}{x}\right)\cdot \frac{2(-1/2+\mu+\delta-\gamma/2)}{(1-2\mu-2\tau)^2}
\\ & = -2H(x)+(2x-1)\log\left(\frac{1-x}{x}\right)
\\ & = \log \left(x(1-x)\right).
\end{align*}
Together, solving for $y$ using \cref{eq:x-y-relation}, we obtain
\begin{align*}
\frac{d[\mc{E}+F_\rho]}{d\tau}\bigg\vert_{\gamma=\gamma_\star(\tau)}
& = -2\log(1-y)+\log 4+\log (x(1-x))+\log\left(\frac{\rho}{1-\rho}\right)+2(4\mu-1)\log\left|\frac{\sin(\rho\pi)}{\rho\pi}\right|
\\ & = \log\left[\left(|\beta|\sqrt{\frac{x}{1-x}}+2\right)^2\left|\frac{\sin(\rho\pi)}{\rho\pi}\right|^{2(4\mu-1)}\frac{\rho}{1-\rho}x(1-x)\right].
\end{align*}
In other words, recalling the definition of $x=x(\tau, \gamma_\star(\tau))$ and $\beta$, we have that
\begin{equation}\label{eq:tau-derivative}
\exp\left(\frac{d\left[E_\rho +F_\rho\right]}{d\tau}\right) = \left|\frac{\sin(\rho\pi)}{\rho\pi}\right|^{2(4\mu-1)} f_\rho(x)\quad\text{where}\quad f_\rho(x)\coloneqq \left(\frac{|1-2\rho|}{\sqrt{\rho(1-\rho)}}\sqrt{\frac{x}{1-x}}+2\right)^2 \frac{\rho}{1-\rho}x(1-x).
\end{equation}
For intuition, we plot $f_\rho(x)$ for various values of $\rho$ in \cref{fig:tau}.
Formally, we break the analysis up into two cases, depending on how $\rho$ compares with $1/2$. 
\\

\textbf{Case 1: }
If $\rho\le 1/2$, then a direct computation yields
\[
f_\rho'(x)
=
2(1-x)
\left(\sqrt{\frac{\rho}{1-\rho}}+\sqrt{\frac{x}{1-x}}\right)
\left(\frac{1-2\rho}{1-\rho}\sqrt{\frac{x}{1-x}}+2\sqrt{\frac{\rho}{1-\rho}}\right)\left(1-\sqrt{\frac{\rho x}{(1-\rho)(1-x)}}\right).
\]
All factors are positive except possibly the last one, whose sign is the sign of $(1-\rho)-x$. Hence, $f_\rho$ is uniquely maximized at $x=1-\rho$ and we check that $f_\rho(1-\rho)=1$.
As $|\sin(\rho\pi)|\le \rho \pi$, \cref{eq:tau-derivative} is less than one, meaning $E_\rho +F_\rho$
is decreasing in $\tau$, so the maximizing $\tau$ is $\tau_\star =0$.
\\

\textbf{Case 2: }
If $\rho> 1/2$, then we similarly compute
\[
f_\rho'(x)
=
2(1-x)
\left(\sqrt{\frac{\rho x}{(1-\rho)(1-x)}}+1\right)
\left(\frac{2\rho-1}{1-\rho}\sqrt{\frac{x}{1-x}}+2\sqrt{\frac{\rho}{1-\rho}}\right)\left(\sqrt{\frac{\rho}{1-\rho}}-\sqrt{\frac{x}{1-x}}\right).
\]
All factors are positive except possibly
the last one, whose sign is the sign of $\rho-x$.
Hence, $f_\rho$ is uniquely maximized at $x=\rho$ and increasing on $[0, \rho]$.
We see from the green line in \cref{fig:phase} that in the regime where \cref{thm:main-biased} is nontrivial, we must have the following numerical inequalities
\begin{equation}\label{eq:numerics}
\frac{1}{2}\le \rho \le 0.67 \quad\text{and}\quad
\mu\ge \bar{\mu}(\rho)\coloneq 0.31+\frac{\rho-0.5}{12},
\end{equation}
since otherwise there is no improvement from $\on{SCL}_\rho(\mu)$. Thus, we can without loss of generality assume these bounds.
Recall from \cref{eq:def-x-y} that
\[x(\tau, \gamma_\star(\tau))\le x(\tau, 0)=\frac{\delta-\tau}{1-2(\mu+\tau)}.
\]
Now, we use the facts that $\mu+\delta\le 1-\rho$ and $\rho > 1/2$ to bound
\[\frac{d}{d\tau}[x(\tau, 0)] = \frac{2\delta+2\mu-1}{(1-2\mu-2\tau)^2}\le \frac{2(1-\rho-\mu)+2\mu-1}{(1-2\mu-2\tau)^2}=\frac{1-2\rho}{(1-2\mu-2\tau)^2} < 0.\]
Therefore, by monotonicity and the fact that $\mu+\delta\le 1-\rho$, we bound
\[0\le x(\tau, \gamma_\star(\tau))\le x(\tau, 0) \leq x(0, 0)=\frac{\delta}{1-2\mu}\le \frac{1-\rho-\mu}{1-2\mu} = \frac{(1-\mu)(1-2\rho)}{1-2\mu} + \rho \coloneqq g_\rho(\mu).\]
For $\mu\le 1/2 < \rho$, note that $g_\rho(\mu) < \rho$, so $f_\rho(x)$ is increasing for $x\in [0, g_\rho(\mu)]$ for every $\mu$.
We also have that $g_\rho(\mu)$ is decreasing in $\mu$ when $\rho > 1/2$, since
$
g'_\rho(\mu)={(1-2\rho)}{(1-2\mu)^{-2}}< 0$.
Together, by monotonicity, we have for $x=x(\tau, \gamma_\star(\tau))$ and $\mu$ assumed to satisfy \cref{eq:numerics},
\begin{equation}
f_\rho(x)\le f_\rho(g_\rho(\mu)) \le f_\rho(g_\rho(\bar{\mu}(\rho))).
\end{equation}
Combining with $|\sin(\rho \pi)|\le \rho\pi$ and \cref{eq:numerics}, we can bound \cref{eq:tau-derivative} above via
\begin{equation}\label{eq:tau-bound-two}
\exp\left(\frac{d[E_\rho+F_\rho]}{d\tau}\right)\le \left|\frac{\sin(\rho\pi)}{\rho\pi}\right|^{2(4\bar{\mu}(\rho)-1)} f_\rho(g_\rho(\bar{\mu}(\rho))) \coloneqq \bar{f}(\rho).
\end{equation}
Observe from \cref{fig:tau} that $\bar{f}(\rho)$ is maximized for $\rho\in [1/2, 0.67]$ at $\rho_\star= 0.56\ldots $ and $\bar{f}(\rho)\le \bar{f}(\rho_\star)=0.9927\ldots <1$. Hence, $E_\rho +F_\rho$ is decreasing in $\tau$, so the maximizing $\tau$ is $\tau_\star =0$.
\begin{figure}
    \centering
 \includegraphics[width=0.9\textwidth]{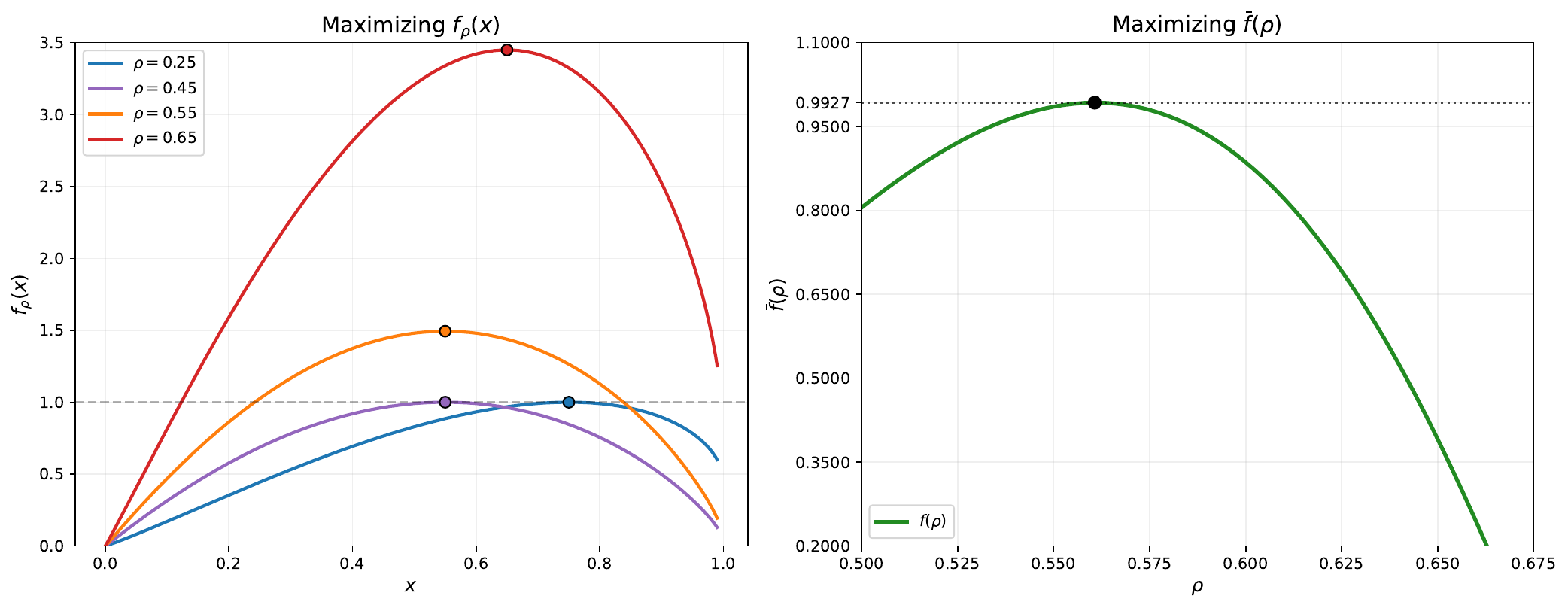}
    \caption{On the left, we plot $f_\rho(x)$ for various values of $\rho$ and observe it is maximized at $1-\rho$ if $\rho\le 1/2$ and at $\rho$ if $\rho\ge 1/2$. On the right, we plot $\bar{f}(\rho)$ and observe it is maximized for $\rho\in [0.5, 0.67]$ at $\bar{f}(0.56)\approx 0.9927$.}
    \label{fig:tau}
\end{figure}

All together, in both cases, under the assumption of \cref{thm:unset-tau-biased}, the maximizer $\tau$ is attained at $\tau_\star =0$. This simplification reduces \cref{thm:unset-tau-biased} to \cref{thm:main-biased}, and the latter is proved.
\end{proof}

\bibliographystyle{alpha}
\bibliography{ref.bib}
\end{document}